\documentclass[11pt]{article}
\usepackage{graphicx}
\usepackage{amsmath}
\usepackage{amsthm}
\usepackage{amssymb}
\usepackage{braket}
\usepackage[margin=1in]{geometry}
\usepackage{xcolor}
\usepackage{algorithm} 
\usepackage[noend]{algpseudocode} 
\usepackage[backref=page,colorlinks,citecolor=blue,bookmarks=true]{hyperref}
\usepackage[capitalise,noabbrev,nameinlink]{cleveref}
\usepackage{dirtytalk}
\usepackage{enumerate}
\usepackage{mathdots}
\usepackage{mathtools}
\usepackage[toc,page]{appendix}
\usepackage{soul}
\usepackage{svg}
\usepackage{xfrac}
\usepackage{thmtools}
\usepackage{float}
\usepackage{authblk}

\makeatletter
\renewcommand\AB@authnote[1]{\rlap{\textsuperscript{\normalfont#1}}}

\makeatother

\DeclareFontFamily{U}{mathx}{}
\DeclareFontShape{U}{mathx}{m}{n}{<-> mathx10}{}
\DeclareSymbolFont{mathx}{U}{mathx}{m}{n}
\DeclareMathAccent{\widehat}{0}{mathx}{"70}
\DeclareMathAccent{\widecheck}{0}{mathx}{"71}

\newtheorem{thm}{Theorem}
\newtheorem{prop}[thm]{Proposition}
\newtheorem{lem}[thm]{Lemma}
\newtheorem{cor}[thm]{Corollary}
\newtheorem{defn}[thm]{Definition}
\newtheorem{prob}{Problem}
\newtheorem{rem}{Remark}

\usepackage{tikz}
\usetikzlibrary{quantikz2}

\newcommand{\psd}[1]{\mathbf{S}_{+}^{#1}}
\newcommand{\diag}{\mathrm{diag}}
\newcommand{\toep}{\mathrm{toep}}

\DeclareMathOperator{\Tr}{Tr}
\DeclareMathOperator{\tr}{tr}
\DeclareMathOperator{\rank}{rank}

\DeclarePairedDelimiter\ceil{\lceil}{\rceil}

\DeclarePairedDelimiter\abs{\lvert}{\rvert}

\crefname{prob}{Problem}{Problems}
\crefname{prop}{Proposition}{Propositions}
\crefname{thm}{Theorem}{Theorems}

\crefname{ineq}{Inequality}{Inequalities}
\creflabelformat{ineq}{#2{\upshape(#1)}#3} 

\crefname{alg}{Algorithm}{Algorithms}
\creflabelformat{ln}{#2{\upshape(#1)}#3} 

\crefname{ln}{Line}{Lines}
\creflabelformat{ln}{#2{\upshape(#1)}#3} 

\renewcommand{\backref}[1]{}

\renewcommand{\backrefalt}[4]{%
\ifcase #1 %
\or
[p.\ #2]%
\else
[pp.\ #2]%
\fi}

\let\originalleft\left
\let\originalright\right
\renewcommand{\left}{\mathopen{}\mathclose\bgroup\originalleft}
\renewcommand{\right}{\aftergroup\egroup\originalright}

\makeatletter
\def\@cite#1#2{\textup{[{#1\if@tempswa , #2\fi}]}}
\makeatother

\title{Translation-Invariant Quantum Algorithms \\ for Ordered Search are Optimal}
\author[1,*]{Joseph Carolan}
\author[1,**]{Andrew M.~Childs}
\author[1,\dag]{Matt Kovacs-Deak}
\author[2,\ddag]{Luke Schaeffer}
\affil[1]{Department of Computer Science, Institute for Advanced Computer Studies, and Joint Center for Quantum Information and Computer Science, University of Maryland\vspace{2pt}}
\affil[2]{David R.\ Cheriton School of Computer Science and Institute for Quantum Computing, University of Waterloo\vspace{2pt}}
\affil[ ]{
$^{*}$jcarolan@umd.edu \quad
$^{**}$amchilds@umd.edu \quad
$^{\dag}$kovacs@umd.edu \quad
$^{\ddag}$lschaeff@uwaterloo.ca
}

\date{March 26, 2025}
\begin{document}

\maketitle
\begin{abstract}
    Ordered search is the task of finding an item in an ordered list using comparison queries. The best exact classical algorithm for this fundamental problem uses $\lceil \log_{2}{n}\rceil$ queries for a list of length $n$. Quantum computers can achieve a constant-factor speedup, but the best possible coefficient of $\log_{2}{n}$ for exact quantum algorithms is only known to lie between $(\ln{2})/\pi \approx 0.221$ and $4/\log_{2}{605} \approx 0.433$. We consider a special class of translation-invariant algorithms with no workspace, introduced by Farhi, Goldstone, Gutmann, and Sipser, that has been used to find the best known upper bounds. First, we show that any bounded-error, $k$-query quantum algorithm for ordered search can be implemented by a $k$-query algorithm in this special class. Second, we use linear programming to show that the best exact $5$-query quantum algorithm can search a list of length $7265$, giving an ordered search algorithm that asymptotically uses $5 \log_{7265}{n} \approx 0.390 \log_{2}{n}$ quantum queries.
\end{abstract}

\section{Introduction}

Ordered search is a fundamental computational primitive with broad applications. The problem asks, given a sorted list of $n$ numbers, to find the location of some specified target. In the classical context, ordered search is solved by the standard binary search algorithm. In the decision tree model, any comparison-based algorithm for this problem requires at least $\log_{2}{n}$ comparisons, because the output is $\log_{2}{n}$ bits and each query returns a single bit. Binary search uses $\ceil{\log_{2}{n}}$ comparisons, and therefore is asymptotically optimal, even up to the leading constant.

Quantum algorithms offer a constant-factor speedup, but the exact value of this constant factor remains elusive. Most upper bounds to date have been achieved by exhibiting an algorithm that uses a small number $k$ of queries to search an $m$-element list exactly, and recursively applying that algorithm to obtain a $k \log_{m}{n}$ upper bound. Attention has focused on a special class of algorithms called \emph{translation-invariant algorithms} introduced by Farhi, Goldstone, Gutmann, and Sipser \cite{fggs99}. These algorithms use no workspace, and are restricted to non-query operations that act diagonally in the Fourier basis, making them easier to analyze. \cite{fggs99} exhibited a 3-query translation-invariant algorithm for searching a $52$-element list exactly, which can be recursively applied to give a $3 \log_{52}{n} \approx 0.526 \log_{2}{n}$-query algorithm. Using a different technique called \emph{pebbling}, H{\o}yer, Neerbek, and Shi gave a $\log_{3}{n}\approx 0.631 \log_{2}{n}$-query algorithm \cite{hns02}. Brookes, Jacokes, and Landahl used gradient descent to exhibit a better translation-invariant algorithm, improving the upper bound to $4 \log_{550}{n} \approx 0.439 \log_{2}{n}$ \cite{bjl04}. Most recently, Childs, Landahl, and Parillo utilized a semidefinite programming formulation to numerically find a 4-query translation-invariant algorithm for exactly searching a list of size $605$, giving the best known upper bound of $4 \log_{605}{n} \approx 0.433 \log_{2}{n}$ \cite{clp}. Lifting the requirement of an exact algorithm, Ben-Or and Hassidim gave a zero-error quantum algorithm using about $0.323 \log_{2}{n}$ queries in expectation \cite{bh07}, also based on a translation-invariant algorithm.

Limitations on quantum algorithms for ordered search have been established through a sequence of lower bounds. Buhrman and de Wolf showed a lower bound of $\Omega\bigl(\sqrt{\log{n}} / \log{\log{n}}\bigr)$ through a reduction to parity \cite{buhrman99lower}. This was improved to $\Omega\bigl(\log{n} / \log{\log{n}}\bigr)$ by Farhi, Goldstone, Gutmann and Sipser \cite{fggs98}, then to $\frac{1}{12} \log_{2}{n}\approx 0.0833 \log_{2}{n}$ by Ambainis \cite{amb99}. H{\o}yer, Neerbek, and Shi used the adversary method to prove the best known quantum lower bound of $\frac{1}{\pi} (\ln{n} - 1) \approx 0.221 \log_{2}{n}$ \cite{amb00,hns02}. Childs and Lee then proved that this is the best possible adversary lower bound, up to an additive constant, by exactly characterizing the adversary quantity of ordered search \cite{ct08}.

\subsection*{Our Results}

As mentioned above, most upper bounds for the ordered search problem have been obtained by considering the subclass of translation-invariant algorithms \cite{fggs98}. Our main result shows that there is no loss of generality in considering this subclass.

\begin{restatable}{thm}{thmMainResult}\label{thm:main_result}
    Let $\varepsilon \geq 0$, and suppose that there is an $\varepsilon$-error quantum algorithm for the $n$-element ordered search problem using $k$ queries. Then there exists an $\varepsilon$-error translation-invariant algorithm solving the same problem using the same number of queries. 
\end{restatable}

Our proof works by studying the symmetries of the ordered search problem. First, we show that controlled access to the query oracle is unnecessary for this problem. Next, we consider a semidefinite program (SDP) of Barnum, Saks, and Szegedy that captures the quantum query complexity of an arbitrary problem \cite{bss}. We show that the symmetries of the ordered search problem can be used to reduce this SDP to another one that characterizes $\varepsilon$-error translation-invariant quantum algorithms. It is worth emphasizing that the definition of a translation-invariant algorithm as in \cite{fggs99} assumes that the algorithm uses no workspace, in addition to having translation symmetry. This means our proofs also establish that no workspace is needed by optimal algorithms for ordered search.

Our second result is an improved upper bound of $5 \log_{7265}{n} \approx 0.390 \log_{2}{n}$ for the exact ordered search problem, which follows from an exact $5$-query quantum algorithm that can search a $7265$-element list. To find this algorithm, we consider a linear programming relaxation of the characterization of \cite{fggs99}, and develop a heuristic approach to finding the algorithm by solving instances of this linear program.

\begin{thm}[Informal]
    The largest list exactly searchable by a $5$-query quantum algorithm is of size $7265$.
\end{thm}

As before, a recursive application of this algorithm leads to the claimed $5 \log_{7265}{n}$ asymptotic upper bound. 

\section{Preliminaries}
In this section we cover some of the technical background and basic notation used throughout this paper. 

\subsection{Notation}
We denote by $J_n$ the $n \times n$ all-ones matrix. We use zero-based numbering when indexing matrices by the integers. Given a matrix $A$ we denote its $(i,j)$ entry by $A[i,j]$. We denote by $A \odot B$ the Hadamard (i.e., element-wise) product of two matrices $A$ and $B$ of the same dimensions. Given an $n \times n$ square matrix $A$ and an integer $-n < j < n$, we denote by $\tr_j(A)$ the trace along the $j$th sub- or superdiagonal of $A$. Concretely, 
\begin{align}
    \tr_j(A) \coloneqq \begin{cases}
        \sum_{i=0}^{n-1-j} A[i,i+j], &j \ge 0 \\ 
        \sum_{i=0}^{n-1+j} A[i-j,i], &j< 0.
    \end{cases} 
\end{align}
We refer to the quantities $\tr_j(A)$ for $-n < j < n$ collectively as the \emph{generalized traces} of the matrix $A$.
For convenience, we let $\tr_j(A) \coloneqq 0$ for $|j| \ge n$. We denote by $\Tr_j(A)$ the $j$th cyclic trace of $A$, i.e.,
\begin{align}
    \Tr_j(A) &\coloneqq \begin{cases}
        \tr_j(A) + \tr_{j-n}(A), &0 \leq j < n \\ 
        \tr_j(A) + \tr_{j+n}(A), &-n < j < 0.
    \end{cases}
\end{align}

Let $[n]$ denote the set $\{1,\ldots,n\}$. We denote permutations $\sigma \in S_n$ using cycle notation. Given two permutations $\sigma, \tau \in S_n$ we write $\sigma\tau$ for the composition $\sigma\circ \tau$. To each $\sigma \in S_n$ we associate the permutation matrix $P_{\sigma}$ obtained by permuting the rows of the identity matrix $I_n$ according to $\sigma$. Because we use zero-based indexing throughout this paper, we regard elements of the symmetric group $S_n$ as permutations of the set $\{0,\ldots,n-1\}$, rather than the more conventional choice of $[n]$.

Given an $n$-tuple of scalars $a = (a_0, \ldots, a_{n-1})$ we denote by $\toep(a_0, \ldots, a_{n-1})$ the $n \times n$ symmetric Toeplitz matrix with first row given by $a$. We denote the set of $n \times n$ positive semidefinite matrices by $\psd{n}$.

We denote by $\mathbb{C}_n[z]$ the vector space of complex polynomials in $z$ that have degree at most $n$. For a polynomial $p(z) = \sum_j c_j z^j \in \mathbb{C}[z]$, we denote the $l_2$-norm of $p$ by $\left\Vert p \right\Vert_2$, i.e., 
\begin{align}
    \left\Vert p \right\Vert_2 = \sqrt{\sum_j \abs{c_j}^2}.
\end{align} 
Whenever the degree of $p(z)$ is understood to be $n$, we denote the coefficient vector of $p(z)$ by $\mathbf{p} \coloneqq \begin{pmatrix} c_0 & \cdots & c_{n} \end{pmatrix}^{\top} \in \mathbb{C}^{n+1}$.

For $b \in \{0,1\}$ we denote by $b^{j}$ the string $b\ldots b$ consisting of $j$ identical bits. Given two bitstrings $w_1$ and $w_2$ we denote their concatenation by $w_1 w_2$.

Throughout this paper we denote by $T$ the cyclic translation operation
\begin{align}
    T\ket{j} &= \ket{j+1} \quad \quad (\forall j \in \mathbb{Z}/2n)
\end{align}
where the arithmetic on the right-hand side is done modulo $2n$. As a matrix, $T$ is just the permutation matrix corresponding to the cycle $(0\;1\;\ldots\;2n-1)$, and acts by mapping $(\begin{matrix}v_0 & v_1 & \cdots & v_{2n-1}\end{matrix})^{\top}$ to $(\begin{matrix}v_{2n-1} & v_0 & \cdots & v_{2n-2}\end{matrix})^{\top}$. Abusing notation slightly, given a bitstring $w = w_0 w_1 \ldots w_{2n-1}$ we write $Tw$ as a shorthand for $w_{2n-1} w_0 \ldots w_{2n-2}$.

\subsection{The Ordered Search Problem}\label{sec:osp}

In the ordered search problem, we are tasked with finding the first element greater or equal to some target value $t$ in a sorted list of numbers $a_0 \leq a_1 \leq \cdots \leq a_{n-1}$, or the last element if no such value exists. We consider this problem in the so-called comparison model, where the list can only be accessed using comparisons that determine whether $a_i \geq t$. In this model, the problem can be framed as a query problem, where each instance is represented by a bitstring $x \in \{0,1\}^n$ such that $x_i = 1$ if and only if $a_i \geq t$. Given query access to $x$, the goal is to output the (zero-based) index of the first $1$ in $x$.  

\begin{prob}[Ordered Search]\label{prob:OSP}
    Let $x = 0^j1^{n-j}$ for some $j \in \{0, \dots, n-1\}$. Given query access to $x$, return $j$.
\end{prob}

Reference~\cite{fggs99} considered a symmetrized variant of this problem by replacing the string $x$ with 
\begin{align}\label{eq:osp_to_symmetrized_osp}
    y &= x \bar{x}    
\end{align}
where $\bar{x}$ denotes the bitwise negation of $x$. The benefit of this change is that the symmetrized input satisfies the translation equivariance
\begin{align}
    (T^{j} y)_i &= y_{i-j} \quad\quad (\forall i, j \in \mathbb{Z}/2n) \label{eq:translation_equivariance}
\end{align}
where the arithmetic is done modulo $2n$. Note that we index bitstrings from $0$, allowing us to do modular arithmetic for the indices as above.

\begin{prob}[Ordered Search, symmetrized]\label{prob:osp-symmetrized}
    Let $w = 1^n0^n$ and $y = T^{j} w$ for some $j \in \mathbb Z / 2n$. Given query access to $y$, return $j \bmod n$.
\end{prob}

It is clear that \cref{prob:OSP,prob:osp-symmetrized} are equivalent, as black-box reductions turn one into the other with no query overhead. Note that in \cref{prob:osp-symmetrized} we only require the answer to be returned modulo $n$, while $j$ ranges over $\mathbb{Z}/2n$. In the reduction from \cref{prob:OSP} to \cref{prob:osp-symmetrized}, the value of the shift modulo $n$ is all that is required to solve the original ordered search instance. In fact, using the phase oracle defined in \cref{sec:querymodel}, it is impossible to distinguish a bitstring from its bitwise complement. This corresponds to distinguishing a shift $j$ from a shift $(j+n) \bmod 2n$. The algorithms considered in this work use a phase oracle of this form.

\subsection{Quantum Query Algorithms}\label{sec:querymodel}
Let $y = T^{j} w$ as in \cref{prob:osp-symmetrized}. In the quantum query model, access to the Boolean oracle $j \mapsto y_j$ is customarily provided by a unitary \emph{phase oracle} acting as 
\begin{align}
    O_y \ket{j} &= (-1)^{y_j} \ket{j}.
\end{align}
While the usual model allows for controlled queries (or equivalently, phase queries with the possibility of a null query), in \cref{subsec:removing_controlled_access} we show that for the ordered search problem we can restrict our attention (with no loss of generality) to algorithms without controlled query access. 

A quantum query algorithm $\mathcal{A}$ that makes $k$ queries to the oracle $O_y$ is defined by a sequence $U_1, U_2, \ldots, U_k$ of unitary operators and an initial state $\ket{\psi_0}$. The final state of such an algorithm on input $y$ is 
\begin{align}
    \ket{\psi_k^{(y)}} = U_k O_y U_{k-1} \dots U_1 O_y \ket{\psi_0}. \label{eq:quantum_query_algorithm}
\end{align}

In the context of our problem, we say that $\mathcal{A}$ solves the ordered search problem \emph{with error $\varepsilon$} if there exists a family $\Pi_0, \dots, \Pi_{n-1}$ of orthogonal projectors with pairwise orthogonal supports such that
\begin{align}
        \left\| \Pi_{j \bmod n} \ket{\psi_k^{(T^{j} w)}} \right\|^2 &\geq 1 - \varepsilon \quad\quad (\forall j \in \{0, 1, \dots, 2n-1\}).
\end{align}

\subsection{Translation-Invariant Algorithms}\label{subsec:translation_invariant_algorithms_prelims}
We work with the symmetrized version of the ordered search problem (\Cref{prob:osp-symmetrized}), which is invariant under cyclic shift. It is natural to expect that cyclic symmetry in the problem will lead to cyclic symmetry in the algorithm. In particular, by symmetrizing the solution of a semidefinite program for quantum query algorithms presented in Ref.~\cite{bss}, one can show that any quantum query algorithm for a symmetric query problem can be chosen to be symmetric, in a certain sense.

Reference~\cite{fggs99} focuses on \say{translation-invariant} quantum algorithms, but in a narrower sense: the unitaries $U_1, \ldots, U_k$ of the algorithm act \emph{only} on the index (from $0$ to $2n-1$), and commute with the translation operator $T$. In other words, the algorithm cannot use any workspace as a side-effect of how the symmetry is defined. More formally, we consider the following class of algorithms.

\begin{defn}[Translation-Invariant Algorithm]\label{defn:translation_invariant_alg}
A $k$-query quantum algorithm $\mathcal{A}$ for the ordered search problem defined by unitaries $U_1, U_2, \ldots, U_k$, as in \cref{eq:quantum_query_algorithm}, is said to be \emph{translation-invariant} if it satisfies the following properties:
\begin{enumerate}[(i)]
    \item $\mathcal{A}$ operates on a single register of dimension $2n$.
    \item The initial state is the uniform superposition $\ket{\psi_0} = \frac{1}{\sqrt{2n}} \sum_{j=0}^{2n-1} \ket{j}$.
    \item Conjugation by the translation operator $T$ leaves the unitaries $U_t$ invariant:
    \begin{align}
        T U_t T^{-1} &= U_t \quad\quad (\forall t\in[k]).\label{eq:translation_invariant_unitary}
    \end{align}
    \item Measurement is made in the basis $\bigl\{\ket{\phi_j}: j = 0, \ldots, n-1 \bigr\}$ where
    \begin{align}
        \ket{\phi_j} = \frac{1}{\sqrt{2}} \bigl(\ket{j} + (-1)^k \ket{j+n}\bigr).
    \end{align}
\end{enumerate}
\end{defn}

The measurement basis depends on the parity of $k$ because in the momentum basis, all the mass shifts from even to odd states (and vice versa) under the query oracle.

Recall that one of our main results (\Cref{thm:main_result}) reduces any exact algorithm to a translation-invariant form. We stress that this is not just about symmetrizing the algorithm, which is possible with existing techniques, but \emph{removing the workspace}. Regrettably, the established name \say{translation-invariant} for this kind of algorithm hides this aspect of our result.

\subsection{Nonnegative Laurent Polynomials}\label{sec:laurent-polynomials}
A degree-$d$ \emph{Laurent polynomial} is a polynomial in $z$ and $z^{-1}$ of the form
\begin{align}
    q(z) &= \sum_{j=-d}^{d} a_j z^j,
\end{align}
where $a_j \in \mathbb{C}$. The Laurent polynomial $q(z)$ is said to be \emph{real-valued} if it is real-valued on the unit circle ($|z|=1$). It is easy to see that $q(z)$ is real-valued if and only if $a_j = \bar{a}_{-j}$. We say that $q(z)$ is \emph{symmetric} if $a_j = a_{-j}$. Thus, a symmetric real-valued Laurent polynomial can be parameterized by $d+1$ real parameters $a_j$ as 
\begin{align}
	q(z) &= \frac{1}{2} \sum_{j=0}^{d} a_j (z^j + z^{-j}).
\end{align}
Lastly, we say that $q(z)$ is \emph{nonnegative} if it is real-valued and nonnegative on the unit circle.

Besides their canonical expression as polynomials in $z$ and $z^{-1}$, we consider two additional ways of representing nonnegative Laurent polynomials. The first is via the Fejér-Riesz theorem, which states that any nonnegative Laurent polynomial can be factorized as the square modulus of an ordinary polynomial \cite{Fejer1916,Riesz1916}.

\begin{thm}[Fejér-Riesz]
    Let $q(z)$ be a degree-$d$ Laurent polynomial. Then $q(z)$ is nonnegative (on the unit circle) if and only if there exists a degree-$d$ polynomial $p(z) \in \mathbb{C}[z]$ such that
    \begin{align}
        q(z) &= \abs{p(z)}^2 \quad (\forall z\colon \abs{z} = 1).
    \end{align}
    Furthermore, $p(z)$ may be chosen so that all of its zeros are contained in the closed unit disk.
\end{thm}

In this context, we refer to the polynomial $p(z)$ as a \emph{spectral factor of $q(z)$}. In general there are several distinct spectral factors $p(z)$ for any nonnegative polynomial $q(z)$. However, $p(z)$ is unique (up to a complex phase) subject to the constraint that all of its zeros are in the closed unit disk \cite{Fejer1916}.

Another representation is the so-called \emph{Gram matrix representation} \cite{dumitrescu}. Recall that a Laurent polynomial $q(z) = \sum_{j=-d}^{d} a_j z^j$ is real-valued if and only if $a_j = \bar{a}_{-j}$ for all $j \in \{-d,\ldots,d\}$. An equivalent condition is that there exists $(d+1) \times (d+1)$ Hermitian matrix $Q$ such that
\begin{align}
    q(z) &= \psi(z^{-1})^{\top} Q \psi(z). \label{eq:hermitian-mx-for-real-laurent-polynomial}
\end{align}
where $\psi(z) = \begin{pmatrix} 1 & z & \cdots & z^{d} \end{pmatrix}^{\top}$.
Equivalently, the coefficients of $q(z)$ are given by generalized traces of $Q$:
\begin{align}
    a_j &= \tr_{j}(Q) \quad\quad (j \in \{-d, \ldots, d\}). \label{eq:gram-matrix-rep-trace-condition}
\end{align}
Then, $q(z)$ is nonnegative if and only if it can be represented this way by a positive semidefinite matrix $Q\succeq 0$.

\begin{lem}[{\cite[Thm.~2.5]{dumitrescu}}]
    Let $q(z) = \sum_{j=-d}^{d} a_j z^j$ be a degree-$d$ Laurent polynomial. Then $q(z)$ is nonnegative (on the unit circle) if and only if there exists a $(d+1)\times (d+1)$ Hermitian, positive semidefinite matrix $Q$ such that $a_j = \tr_{j}(Q)$ for $j \in \{-d, \ldots, d\}$.
\end{lem}

One example of a symmetric, nonnegative Laurent polynomial is the \emph{Fejér kernel}
\begin{align}
    F_n(z) \coloneqq \sum_{j = -(n-1)}^{(n-1)} \Bigl(1 - \frac{\abs{j}}{n}\Bigr) z^j. \label{eq:fejer-kernel}
\end{align}
It is easy to check that the polynomial $p(z) \coloneqq {\frac{1}{\sqrt{n}} (1 + z + \cdots + z^{n-1})}$ is a spectral factor of $F_n$. Moreover, the outer product $\mathbf{p}^{*} \mathbf{p}^{\top} = \frac{1}{n} J_n$ is a Gram matrix representation of $F_n$, where $\mathbf{p} \in \mathbb{C}^n$ is the column vector whose entries are the coefficients of $p(z)$.

Given a spectral factor $p(z)$ of a nonnegative Laurent polynomial $q(z)$, the rank-$1$ projector $\mathbf{p}^{*}\mathbf{p}^{\top}$ is a Gram matrix representation of the polynomial $q(z)$. A key step in this work exploits a partial converse of this fact.

\begin{thm}\label{thm:maximizing-top-left-entry}
    Let $q(z) = \sum_{j=-(n-1)}^{n-1} a_j z^j$ be a nontrivial, nonnegative Laurent polynomial. Suppose that $Q \in \mathbf{S}_{+}^{n}$ is a Gram matrix representation of $q(z)$ for which the top left entry $Q[0,0]$ is maximal. Then $Q$ has rank $1$. Furthermore, $Q = \mathbf{p}^{*}\mathbf{p}^{\top}$ where $\mathbf{p} \in \mathbb{C}^n$ is the coefficient vector of a spectral factor $p(z)$ of $q(z)$.
\end{thm}

For completeness, we present the proof given in \cite{dumitrescu}.
\begin{proof}
    We first show that the rank of $Q$ is $1$. 
    For each $-n < j < n$, the set of matrices $M \in \mathbb{C}^{n\times n}$ with $\tr_j{M} = a_j$ is closed. Thus so is the set of Gram matrix representations of $q(z)$, being the intersection of finitely many closed sets. Therefore, there exists a Gram representation $Q \in \mathbf{S}_{+}^n$ of $q(z)$ with maximal top left entry. As $q(z)$ is nontrivial, $Q \neq 0$. We may express $Q$ as 
    \begin{align}
        Q &= \begin{pmatrix}
            \alpha & v^{\dagger} \\ 
            v & \widehat{Q}
        \end{pmatrix},
    \end{align}
    where $\alpha$ is a positive real scalar. Since $Q$ is positive semidefinite, so are all of its conjugates. In particular,
    \begin{align}
        \begin{pmatrix}
            1 & 0 \\ 
            -\frac{v}{\alpha} & I_{n-1}
        \end{pmatrix} Q \begin{pmatrix}
            1 & 0 \\ 
            -\frac{v}{\alpha} & I_{n-1}
        \end{pmatrix}^{\dagger} &= \begin{pmatrix}
            \alpha & 0 \\
            0 & \widehat{Q} - \frac{vv^{\dagger}}{\alpha}
        \end{pmatrix} \succeq 0.
    \end{align}
    Thus the $(n-1)\times (n-1)$ submatrix $P \coloneqq \widehat{Q} - \frac{vv^{\dagger}}{\alpha}$ is positive semidefinite as well, and so is
    \begin{align}
        Q' &\coloneqq \begin{pmatrix}\alpha & v^{\dagger} \\ v & \frac{vv^{\dagger}}{\alpha}\end{pmatrix} + \begin{pmatrix}P & 0 \\ 0 & 0 \end{pmatrix}.
    \end{align} 
    Since $Q' = Q + \diag(P, 0) - \diag(0, P)$, all of the generalized traces of $Q$ and $Q'$ are identical, i.e., $Q'$ is another Gram representation of $q(z)$.

    Now suppose for contradiction that $\rank(Q) \neq 1$. Then $P \neq 0$ and in particular $P$ must have a strictly positive diagonal entry. Cyclically rotating the elements of $P$ down and to the right if necessary, we may assume that $P[0,0] > 0$. Note that this change does not affect the fact that $Q'$ is a Gram matrix representation of $q(z)$. Thus, $Q'$ is a Gram representation of $q(z)$ with $Q'[0,0] = Q[0,0] + P[0,0] > Q[0,0]$, giving us the required contradiction. 

    The \say{furthermore} part follows by a straightforward calculation. Since $Q$ is a positive semidefinite matrix of rank $1$, ${Q = \mathbf{p}^{*}\mathbf{p}^{\top}}$ for some nonzero vector $\mathbf{p} = (\begin{matrix}b_{0} & \cdots & b_{n-1}\end{matrix})^{\top}$. Then the $(i,j)$ entry of $Q$ equals $b_i^{*} b_j$, and as $Q$ represents $q(z)$,
    \begin{align}
        a_j &= \tr_{j}(Q) = \begin{cases}
            \sum_{i = 0}^{n-1-j} b_{i}^{*} b_{i+j}, \quad j \geq 0\\ 
            \sum_{i = 0}^{n-1+j} b_{i-j}^{*} b_{i}, \quad j < 0.
        \end{cases}
    \end{align}
    Finally, notice that these are precisely the coefficients of the Laurent polynomial
    \begin{align}
        p(z) p\Bigl(\frac{1}{z^*}\Bigr)^*,
    \end{align}
    where $p(z) \coloneqq \sum_{j=0}^{n-1} b_j z^j$.
\end{proof}

\subsection{Characterizing Algorithms by Polynomials}

Farhi, Goldstone, Gutmann, and Sipser developed a polynomial characterization of translation-invariant quantum algorithms for the ordered search problem. The following result is implicit in their work \cite[Section~4]{fggs99}.

\begin{restatable}{thm}{thmFGGSPolys}\label{thm:fggs-polys}
    There exists an $\varepsilon$-error, $k$-query translation-invariant algorithm for the $n$-element ordered search problem if and only if there exist polynomials $p_1, \ldots, p_k \in \mathbb{C}_{n-1}[z]$ such that
    \begin{align}
        p_0(z) &\coloneqq \frac{1}{\sqrt{n}}\Bigl(1 + z + \cdots + z^{n-1}\Bigr) \label{eq:fggs-polys-initial-step} \\ 
        \abs{p_t(z)} &= \abs{p_{t-1}(z)} &&(\forall z\colon z^n = (-1)^t,\, t \in [k]) \label{eq:fggs-polys-forward-step} \\
        \abs{p_k(0)}^2 &\geq 1-\varepsilon  \label{ineq:fggs-polys-final-ineq} \\ 
        \left\Vert p_t \right\Vert_2 &= 1 &&(t \in [k]). \label{eq:fggs-polys-normalization}
    \end{align}
\end{restatable}

For exact algorithms (corresponding to $\varepsilon = 0$), the authors considered the following equivalent characterization in terms of the nonnegative Laurent polynomials $q_t(z) \coloneqq p_t(z) p_t(1/z^{*})^*$. The main benefit of this characterization is that we no longer need to deal with absolute values, as the constraints in \cref{eq:fggs-polys-forward-step} are replaced by \cref{eq:fggs-pos-polys-forward-step}.

\begin{restatable}[\cite{fggs99}]{thm}{thmFGGSPosPolys}\label{thm:fggs-pos-polys}
    There exists an exact $k$-query translation-invariant quantum algorithm for the $n$-element ordered search problem if and only if the following program is feasible. 
    \begin{align}
    \noalign{\centering{\text{Find symmetric, nonnegative Laurent polynomials $q_0, \ldots, q_k$ of degree less than $n$ such that}}}
        q_0 &\equiv F_n \label{eq:fggs-pos-polys-init-state}\\
        q_t(z) &= q_{t-1}(z) & (\forall z \text{ with } z^n = (-1)^t, 1 \leq t \leq k) \label{eq:fggs-pos-polys-forward-step}\\
        q_k &\equiv 1 \label{eq:fggs-pos-polys-final-state}\\
        \frac{1}{2\pi} \int_0^{2\pi} q_t (e^{i\theta})\, d\theta &= 1 & (0 \leq t \leq k). \label{eq:fggs-pos-polys-norm}
    \end{align}
\end{restatable}
A follow-up work of Childs, Landahl, and Parillo used the Gram matrix representation to turn the program in \cref{thm:fggs-pos-polys} into a semidefinite program that is feasible if and only if an exact, $k$-query translation invariant algorithm exists for the $n$-element ordered search problem \cite[Theorem~3]{clp}. One limitation of their program, inherited from \cref{thm:fggs-pos-polys}, is that it only applies to the exact ($\varepsilon = 0$) case.

In this work, we extend the result of \cite{clp} by constructing an SDP that characterizes the existence of $\varepsilon$-error algorithms for the ordered search problem for all choices of $\varepsilon > 0$.

\begin{restatable}{thm}{thmCLP}\label{thm:clp}
There exists an $\varepsilon$-error translation-invariant, $k$-query quantum algorithm for the $n$-element ordered search problem if and only if the following semidefinite program is feasible:
\begin{align}
    Q^{(0)} &\coloneqq J_n/n \label{eq:clp-sdp-initial-step} \\
    \tr_j{Q^{(t)}} + (-1)^t \tr_{j-n}{Q^{(t)}} &= \tr_j{Q^{(t-1)}} + (-1)^t \tr_{j-n}{Q^{(t-1)}} &(0 < j < n, t \in [k]) \label{eq:clp-sdp-forward-step} \\ 
    Q^{(k)}[0,0] &\geq 1-\varepsilon \label{eq:clp-sdp-final-mx-inequality} \\
    \tr{Q^{(t)}} &= 1 &(t \in [k]) \label{eq:clp-sdp-normalization} \\
    Q^{(t)} &\in \psd{n}  &(t \in [k]) \label{eq:clp-sdp-semidefinite-const}
\end{align}
where $J_n$ denotes the $n \times n$ all-ones matrix.
\end{restatable}
We prove \cref{thm:clp} in \cref{subsec:generalizing-clp}. We use this theorem to show that, informally speaking, (workspace-free) translation-invariant algorithms are in fact optimal for the ordered-search problem. The key idea is that we may without loss of generality assume that the variables of the SDP are all rank-$1$ matrices. While this is a nonconvex constraint in general, it can be assumed in this particular case, as can be shown using \cref{thm:maximizing-top-left-entry}.

\section{Translation-Invariant Quantum Algorithms Are Optimal}\label{sec:invariant-algs-optimal}

As discussed in the introduction, much of the prior research on quantum algorithms for the ordered search problem focused on the special class of algorithms that are (workspace-free and) translation-invariant (\cref{defn:translation_invariant_alg}). Indeed, the best known algorithms in the exact ($\varepsilon = 0$) setting use this framework \cite{fggs99,clp}. In this section we prove that this class of algorithms is in fact optimal in the sense that for any $\varepsilon$-error algorithm that solves the $n$-element ordered search problem, there exists an $\varepsilon$-error translation-invariant algorithm that solves the problem using the same number of queries. 

\subsection{Removing Controlled Access}\label{subsec:removing_controlled_access}
The first step in our argument is to show that controlled query access is not necessary for the symmetrized ordered search problem. We use this to eliminate matrices corresponding to null queries when we consider a semidefinite programming characterization of quantum query complexity in \cref{subsec:translation_invariant_algorithms}.

The key observation is that the problem asks for the shift modulo $n$, which means that the input $x$ and its bitwise negation $\bar{x}$ correspond to the same output.

\begin{prop}\label{prop:no-null-queries}
    Any quantum algorithm for the symmetrized ordered search problem (\cref{prob:osp-symmetrized}) that has controlled access to the input oracle can be simulated by an algorithm without controlled access with the same number of queries and success probability.
\end{prop}

\begin{proof}
    Let $w$ be as in \cref{prob:osp-symmetrized}. Let $\mathcal A$ be an algorithm that makes controlled queries to the input oracle $O$. For an input string $s = T^j w$, a controlled query to $O$ acts as
    \begin{align}
        \textsf{ctrl-}O \ket{b} \ket{x} &\coloneqq (-1)^{b \cdot s_x} \ket{b} \ket{x}.
    \end{align}
    We modify $\mathcal{A}$ to replace each $\textsf{ctrl-}O$ gate with the circuit below, which uses controlled-SWAPs to apply $O$ to either an eigenstate (in this case $\ket{0}$) when $b = 0$, or to $\ket{x}$ when $b=1$.
    
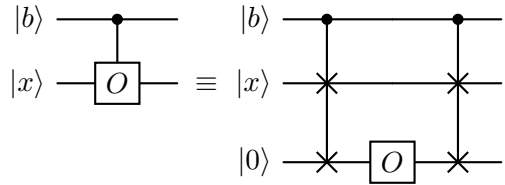
\begin{equation*}
    \begin{quantikz}[align equals at=2]
        \lstick{$\ket{b}$} & \ctrl{1} & \\
        \lstick{$\ket{x}$} & \gate{O} &
    \end{quantikz}
    \equiv 
    \begin{quantikz}[align equals at=2]
        \lstick{$\ket{b}$} & \ctrl{2} & & \ctrl{2} & \\
        \lstick{$\ket{x}$} & \targX{} & \ghost{O} & \targX{} & \\
        \lstick{$\ket{0}$} & \targX{} & \gate{O} & \targX{} & \\
    \end{quantikz}
    \end{equation*}
    In other words, we map
    \begin{align}
        \textsf{ctrl-}O \ket{b} \ket{i} \mapsto \underbrace{\textsf{ctrl-SWAP} \cdot (O \otimes I) \cdot \textsf{ctrl-SWAP}}_{O'} \ket{b} \ket{0} \ket{x} \label{eq:nonull-new-oracle},
    \end{align}
    where 
    \begin{align}
        \textsf{ctrl-SWAP} \ket{b} \ket{x} \ket{y} \coloneqq \begin{cases}
            \ket{b}\ket{x}\ket{y}, & \text{if $b=0$} \\
            \ket{b}\ket{y}\ket{x}, & \text{if $b=1$}. 
        \end{cases}
    \end{align}
    For simplicity, we omit the $\ket{0}$ register of \cref{eq:nonull-new-oracle} going forward. This register remains in the $\ket{0}$ state no matter the values of $\ket{b}$ and $\ket{i}$, so this loses no information. We can write the action of $O'$ as \begin{align}
        O' \ket{b}\ket{i} &= (-1)^{x_i \cdot b}\underbrace{(-1)^{x_0 \cdot (1-b)}}_{\mathsf{Phase}} \ket{b}\ket{i},
    \end{align}
    where the term labeled $\mathsf{Phase}$ is the only difference from the original oracle $O$. For any input of the form $x \in 0^j1^n0^{n-j}$ with $j > 0$, the action of $O'$ is the same as $O$ (neglecting the ancilla register), because $x_0=0$ and therefore the $\mathsf{Phase}$ term is $1$. On such inputs, $\mathcal A'$ outputs the same number as $\mathcal A$, namely $j \bmod n$.

    The remaining inputs are of the form $x=1^j0^n1^{n-j}$ for $j>0$. Such an input is simply the negation of an input of the form $0^j1^n0^{n-j}$, and an algorithm with an uncontrolled phase oracle cannot distinguish the negation as the oracles only differ by a global phase. It follows that $\mathcal A'$ outputs $j \bmod n$ on inputs of this form as well, which is the correct answer.
\end{proof}

\subsection{An SDP Characterization of Translation-Invariant Algorithms}\label{subsec:generalizing-clp}

In this section we establish our semidefinite programming characterization of translation-invariant algorithms for the ordered search problem. Our starting point is the polynomial characterization of \cite{fggs99}, which we restated below. In particular, we show that the SDP stated in our \cref{thm:clp} is feasible if and only if the conditions in \cref{thm:fggs-polys} are satisfied. 

\thmFGGSPolys*

We recall our characterization below.

\thmCLP*

This theorem is reminiscent of the characterization of exact ($\varepsilon=0$) algorithms given by Childs, Landahl, and Parillo \cite{clp}, except that the condition for the last matrix in \cref{eq:clp-sdp-final-mx-inequality} is different, allowing for a nonzero error probability. While this may seem like a minor change, establishing this generalization requires additional ingredients. In particular, our proof crucially relies on the fact that given a nonnegative Laurent polynomial, one can always choose a Gram matrix representation that has rank one (\cref{thm:maximizing-top-left-entry}). 

We begin by establishing two lemmas.

\begin{lem}\label{lem:clp-sdp-alt-forward-step}
    The constraint in \cref{eq:clp-sdp-forward-step} can be replaced by the following equation without affecting the feasibility of the SDP stated in \cref{thm:clp}: 
    \begin{align}
        \sum_{j=-(n-1)}^{(n-1)} \tr_{j}(Q^{(t)}) z^j &= \sum_{j=-(n-1)}^{(n-1)} \tr_{j}(Q^{(t-1)}) z^j &(\forall z\colon\ z^n = (-1)^t, t\in [k]).  \label{eq:clp-sdp-forward-step-alt}
    \end{align}
\end{lem}

\begin{proof}
    First observe that the variables $Q^{(1)}, \ldots, Q^{(k)} \in \mathbf{S}_{+}^n$ of the SDP can be assumed to be real symmetric matrices. Indeed, if $Q^{(1)}, \ldots, Q^{(k)}$ give a feasible solution, then it is easy to see that so do the complex conjugates $\overline{Q}^{(1)}, \ldots, \overline{Q}^{(k)}$ of these matrices. By convexity, the matrices $\widehat{Q}^{(t)} \coloneqq (Q^{(t)} + \overline{Q}^{(t)})/2$ give another solution, and these are real symmetric. Note that the same argument works even if \cref{eq:clp-sdp-forward-step} is replaced by \cref{eq:clp-sdp-forward-step-alt}. Therefore we may restrict our attention to real symmetric matrices. 
    
    We show that under the assumption that $Q^{(1)}, \ldots, Q^{(k)} \in \mathbf{S}_{+}^n$ are real symmetric matrices of trace $1$, the constraints in \cref{eq:clp-sdp-forward-step,eq:clp-sdp-forward-step-alt} are equivalent. For each $t \in \{0, \ldots, k\}$, let $q_t(z)$ be the nonnegative Laurent polynomial represented by $Q^{(t)}$ through the Gram matrix representation. Since the matrix $Q^{(t)}$ is symmetric, so is the polynomial $q_t(z)$. Thus we may parameterize it as 
    \begin{align}
        q_t(z) &= \frac{1}{2} \sum_{j = 0}^{n-1} a_j^{(t)} (z^j + z^{-j}).
    \end{align}
    \Cref{eq:clp-sdp-forward-step} is equivalent to
    \begin{align}
        a^{(t)}_j + (-1)^t a^{(t)}_{n-j} &= a^{(t-1)}_j + (-1)^t a^{(t-1)}_{n-j} &(0 < j < n, t\in [k]) \label{eq:clp-sdp-forward-step-alt-coeff-expr} \\
        \intertext{On the other hand, \cref{eq:clp-sdp-forward-step-alt} is equivalent to the following equation:}
        q_t(z) &= q_{t-1}(z) &(\forall z\colon\ z^n = (-1)^t, t \in [k]). \label{eq:pos-poly-forward-step-restated}
    \end{align} 
    We claim that \cref{eq:clp-sdp-forward-step-alt-coeff-expr,eq:pos-poly-forward-step-restated} are equivalent under the assumption that the matrices $Q^{(t)}$ have trace one. This is essentially shown in \cref{prop:redundant-eqs}, which is proven in \cref{sec:lp-eliminating-variables}. More precisely, \cref{prop:redundant-eqs} shows that \cref{eq:pos-poly-forward-step-restated} is equivalent to
    \begin{align}
        (I_n + (-1)^t V_n)(a^{(t)} - a^{(t-1)}) &= 0 \quad (\forall t \in [k]) \label{eq:clp-sdp-forward-step-alt-mx-exp}
    \end{align}
    where $V_n = \diag(I_1, X_{n-1})$, with $X_{n-1}$ denoting the $(n-1)\times(n-1)$ permutation matrix with $1$s on the anti-diagonal. This equation is almost identical to \cref{eq:clp-sdp-forward-step-alt-coeff-expr}, except that it requires $a^{(t)}_0 = a^{(t-1)}_0$ whenever $t$ is even. 
    However, the constraint that $\tr{Q^{(t)}} = 1$ (for all $t \in [k]$) is equivalent to $a_0^{(t)} = 0$ (for all $t \in [k]$), and under this constraint \cref{eq:clp-sdp-forward-step-alt-mx-exp} is indeed equivalent to \cref{eq:clp-sdp-forward-step-alt-coeff-expr}.
\end{proof}

The next lemma is a corollary of \cref{thm:maximizing-top-left-entry}.

\begin{lem}\label{lem:clp-rank-1-mxs}
    Consider the SDP from \cref{thm:clp} with the constraint in \cref{eq:clp-sdp-forward-step} replaced by \cref{eq:clp-sdp-forward-step-alt}. The variables $Q^{(t)}$ of this SDP can, without loss of generality, be taken to have rank $1$.
\end{lem}

\begin{proof}
    Suppose that $Q^{(1)}, \ldots, Q^{(k)} \in \mathbf{S}_{n}^{+}$ is a feasible solution of the SDP. For each $t \in [k]$, we can associate to $Q^{(t)}$ a nonnegative Laurent polynomial $q_t(z)$ of degree $n-1$ through the Gram matrix representation. By \cref{thm:maximizing-top-left-entry} this polynomial is also represented by some rank-$1$ matrix $\widetilde{Q}^{(t)} \in \mathbf{S}_n^{+}$. In particular, $\widetilde{Q}^{(t)}$ is such that all of the generalized traces of $Q^{(t)}$ and $\widetilde{Q}^{(t)}$ are equal:
    \begin{align}
        \tr_j{\widetilde{Q}^{(t)}} &= \tr_j{Q^{(t)}} \quad (-n < j < n).
    \end{align}
    We conclude that the rank-$1$ matrices $\widetilde{Q}^{(1)}, \ldots, \widetilde{Q}^{(k)}$ satisfy the constraints in \cref{eq:clp-sdp-forward-step-alt,eq:clp-sdp-normalization}. Finally, \cref{thm:maximizing-top-left-entry} guarantees that the top left entry of $\widetilde{Q}^{(k)}$ is maximal among Gram matrix representations of $q_k(z)$, and thus the constraint in \cref{eq:clp-sdp-final-mx-inequality} is also satisfied: 
    \begin{align}
        \widetilde{Q}^{(k)}[0,0] &\geq Q^{(k)}[0,0] \geq 1-\varepsilon.
    \end{align}
\end{proof}

We now have all the ingredients of the proof of \cref{thm:clp}.

\begin{proof}[Proof (\cref{thm:clp}).]
    We show that for any $k, n\in \mathbb{N}$, the conditions of \cref{thm:fggs-polys} are met if and only if the SDP stated in the theorem is satisfied. We consider the SDP with the constraint of \cref{eq:clp-sdp-forward-step} replaced by \cref{eq:clp-sdp-forward-step-alt}, as this change does not affect feasibility (as shown in \cref{lem:clp-sdp-alt-forward-step}).

    We define a map $\mathbb{C}_{n-1}[z]\setminus\{0\} \to \mathbf{S}_{+}^n$ from the set of nontrivial polynomials of degree less than $n$ to the set of $n \times n$ positive semidefinite matrices as follows.
    Given a nonzero polynomial $p(z) = \sum_{j=0}^{n-1} b_j z^j$, define the positive semidefinite matrix $Q \coloneqq \mathbf{p}^{*}\mathbf{p}^{\top}$, where $ \mathbf{p}$ is the $n$-dimensional column vector whose $j$th entry is $b_j$. 
    This is a surjection onto the set of rank-$1$ positive semidefinite matrices that maps two polynomials that differ by a complex phase to the same matrix.

    Given a list $(p_t)_{t=1}^{k}$ of nonzero polynomials $p_t(z) = \sum_{j=0}^{n-1} b^{(t)}_j z^j$, let $(Q^{(t)})_{t=1}^{k}$ be their images under this mapping. We show that the polynomials $(p_t)_{t=1}^{k}$ satisfy the constraints of \cref{thm:fggs-polys} if and only if the matrices $\bigl(Q^{(t)}\bigr)_{t=1}^{k}$ are a feasible solution of the SDP (with the constraint in \cref{eq:clp-sdp-forward-step} replaced by \cref{eq:clp-sdp-forward-step-alt}). We examine the constraints of \cref{thm:fggs-polys} one by one:

    \begin{enumerate}[(i)]
        \item For each $t \in \{0, \ldots, k\}$ define the formal Laurent polynomial 
        \begin{align}
            q_t(z) &\coloneqq p_t(z) p_t\Bigl(\frac{1}{z^{*}}\Bigr)^{*}.
        \end{align} 
        By construction $q_t(z) = \abs{p_t(z)}^2$ for any $z$ on the unit circle, and the constraint in \cref{eq:fggs-polys-forward-step} is equivalent to
        \begin{align}
            q_t(z) &= q_{t-1}(z) \quad (\forall z\colon\ z^n = (-1)^t, 1 \leq t \leq k). \label{eq:pos-poly-forward-step}
        \end{align}
        Note that $q_t(z)$ is nonnegative on the unit circle and $Q^{(t)}$ is a Gram matrix representation of it. In particular, the above constraint can be expressed in terms of the matrices $Q^{(t)}$ as 
        \begin{align}
            \sum_{j=-(n-1)}^{(n-1)} \tr_{j}(Q^{(t)}) z^j &= \sum_{j=-(n-1)}^{(n-1)} \tr_{j}(Q^{(t-1)}) z^j \quad\quad (\forall z\colon\ z^n = (-1)^t, 1 \leq t \leq k).
        \end{align}
        This is just \cref{eq:clp-sdp-forward-step-alt}, as desired.
        \item Observe that the top left entry of $Q^{(k)}$ is the square modulus of the constant term of $p_k(z)$, i.e.,
        \begin{align}
            Q^{(k)}[0,0] = \bigl\vert b^{(k)}_0 \bigr\vert^2 = \left\vert p_k(0) \right\vert^2.
        \end{align}
        Thus, the constraint capturing the success probability (\cref{ineq:fggs-polys-final-ineq}) is equivalent to \cref{eq:clp-sdp-final-mx-inequality}. 
        \item By construction,
        \begin{align}
            \tr{Q^{(t)}} &= \sum_{j=0}^{n-1} Q^{(t)}[j,j] = \sum_{j=0}^{n-1} \bigl\vert b^{(t)}_j\bigr\vert^2 = \Vert p_t \Vert_2^2.
        \end{align}
        Hence the normalization constraint (\cref{eq:fggs-polys-normalization}) is equivalent to \cref{eq:clp-sdp-normalization}.
    \end{enumerate}
    We have shown that the conditions in \cref{thm:fggs-polys} are met if and only if the SDP (with \cref{eq:clp-sdp-forward-step} replaced by \cref{eq:clp-sdp-forward-step-alt}) has a feasible solution that consists entirely of rank-$1$ matrices. This restriction, however, is with no loss of generality. Indeed, by \cref{lem:clp-rank-1-mxs}, whenever this SDP is feasible, there exists a feasible solution that consists of rank-$1$ matrices only. 
    
    We conclude that the SDP stated in the theorem, with \cref{eq:clp-sdp-forward-step} replaced by \cref{eq:clp-sdp-forward-step-alt}, is feasible if and only if the conditions of \cref{thm:fggs-polys} are met. By \cref{lem:clp-sdp-alt-forward-step}, the same is true of the original SDP as well.
\end{proof}

\subsection{Symmetries and the Barnum-Saks-Szegedy SDP}

In this section, we review a semidefinite programming characterization of quantum query complexity developed by Barnum, Saks, and Szegedy \cite{bss}, and study how symmetries of a problem manifest in terms of the solutions of this SDP. 

For $X\subseteq \{0,1\}^n$, let $f\colon X \to Y$ be a function, $k$ be a positive integer, and $\varepsilon \in [0, 1/2)$. Consider the task of deciding whether there exists a $k$-query quantum algorithm that computes $f$ with error at most $\varepsilon$. To this end Barnum, Saks, and Szegedy defined the following a semidefinite program, and showed that the feasibility of this program determines the answer to the question.

\begin{defn}[{\cite[Section 3]{bss}}]\label{defn:bss-sdp}
Let $X \subseteq \{0,1\}^n$, $f\colon X \to Y$ be a function, $k$ be a positive integer, and $\varepsilon \in [0, 1/2)$. 
\begin{itemize}
    \item Let $J \coloneqq J_{\abs{X}}$ denote the $\abs{X} \times \abs{X}$ all-ones matrix.
    \item For $i = 0, \ldots, n-1$, let $E_i$ be the $\abs{X} \times \abs{X}$ matrix given by $E_i[x,x'] = (-1)^{x_i+x'_i}$. 
    \item For each $y \in Y$, let $\Delta_y$ denote the diagonal matrix with $\Delta_y[x,x] = \delta_{f(x), y}$.
\end{itemize}
Let $P(f,k,\varepsilon)$ be the following semidefinite program:
    \begin{align}
      J &= \sum_{i=0}^{n-1} M_i^{(0)} + N^{(0)} \label{eq:bss-initial-step} \\ 
     \sum_{i=0}^{n-1} E_i \odot M_i^{(t)} + N^{(t)} &= \begin{cases}
        \sum_{i=0}^{n-1} M_i^{(t+1)} + N^{(t+1)}, &t = 0, \ldots, k-2 \\ 
        \sum_{y \in Y} \Gamma_y, &t = k-1
    \end{cases} \label{eq:bss_forward_step} \\ 
    \Delta_y \odot \Gamma_y &\geq (1-\varepsilon) \Delta_y \quad (\forall y \in Y) \label[ineq]{eq:bss_entrywise_ineq} \\ 
    M_i^{(t)}, N^{(t)}, \Gamma_y &\in \psd{\lvert X \rvert}.
\end{align}
\end{defn}

\begin{thm}[{\cite[Theorem 1]{bss}}]\label{thm:bss_thm}
    Let $X, k, \varepsilon$ and $f\colon X \to Y$ be as in \cref{defn:bss-sdp}. Then there exists a $k$-query quantum algorithm computing $f$ with error at most $\varepsilon$ if and only if the program $P(f,k,\varepsilon)$ is feasible.
\end{thm}

\begin{rem}\label{rem:bss_gramians} 
    It is also shown implicitly in Ref.~\cite{bss} that the matrices $G^{(t)} \coloneqq \sum_{i = 0}^{n-1} M_{i}^{(t)} + N^{(t)}$ for $t = 0, \ldots, k-1$ can be taken without loss of generality to be Gram matrices of unit-norm vectors. Note that, as is the case for $t = 0$, these vectors are not necessarily distinct.
\end{rem}

We will extensively use symmetries in our study of the ordered search problem. We formalize the notion of a symmetry of a function $f$ as a $3$-tuple of permutations of the input set, index set, and output set, that satisfy certain relations involving $f$.
 
\begin{defn}[Symmetry of a query problem]\label{defn:symmetry}
    Let $X \subseteq \{0,1\}^n$ and let $f\colon X \to Y$ and $(E_i)_{i=0}^{n-1}$ be as in \cref{defn:bss-sdp}. A symmetry of the function $f$ is a 3-tuple $(\pi, \rho, \sigma)$ of permutations of, respectively, the inputs $X$, query locations $\{0, \ldots, n-1\}$, and outputs $Y$ such that 
    \begin{align}
        E_{\rho(i)}[x,x'] &= E_{i}[\pi(x),\pi(x')] &\forall i \in \{0, \ldots, n-1\}, \forall x, x' \in X \label{eq:symmetry_defn_E_mxs} \\ 
        \sigma(f(x)) &= f(\pi(x)) &\forall x \in X. \label{eq:symmetry_defn_mu_and_sigma}
    \end{align}
\end{defn}

In this definition, $\pi$ permutes the inputs, $\rho$ permutes the query locations, and $\sigma$ permutes the outputs. Next we describe a composition operation for symmetries.

\begin{rem}\label{rem:composition_of_symmetries}
Given a function $f\colon X \to Y$, let the matrices $(E_i)_{i=0}^{n-1}$ be as in \cref{defn:symmetry}, and suppose that $(\pi_1, \rho_1, \sigma_1)$ and $(\pi_2, \rho_2, \sigma_2)$ are symmetries of $f$. Then 
\begin{align}
    E_{\rho_1\rho_2(i)}[x,x'] &= E_{\rho_2(i)}[\pi_1(x), \pi_1(x')] = E_{i}[\pi_2\pi_1(x), \pi_2\pi_1(x')] \quad \text{ and } \\ 
    \sigma_2\sigma_1(f(x)) &= \sigma_2(f(\pi_1(x))) = f(\pi_2\pi_1(x)).
\end{align}
Thus we see that the symmetries of $f$ form a group under the composition operation
\begin{align}
    (\pi_2, \rho_2, \sigma_2) \circ (\pi_1, \rho_1, \sigma_1) \coloneqq (\pi_2 \circ \pi_1,\> \rho_1 \circ \rho_2,\> \sigma_2\circ\sigma_1). \label{eq:composition_of_symmetries}
\end{align} 
Note that in the second coordinate, the order of composition is reversed.
\end{rem}

Next we show how symmetries transform the $\Delta$ matrices of \cref{defn:bss-sdp}.
\begin{prop}\label{prop:conjugating_delta_mx}
    Suppose that $(\pi, \rho, \sigma)$ is a symmetry of a function $f\colon X \to Y$, and let the matrices $\Delta_y \in \mathbb{R}^{\abs{X} \times \abs{X}}$ be defined as the diagonal matrices with $\Delta_y[x,x] = \delta_{f(x), y}$. Then for any $y \in Y$,
    \begin{align}
        \Delta_y &= P^{-1}_{\pi} \Delta_{\sigma(y)} P_{\pi}, \label{eq:conjugating_delta_mx}
    \end{align}
    where $P_\pi$ is the permutation matrix associated with $\pi$.
\end{prop}
\begin{proof} First note that the matrix on the right-hand side is also diagonal. The equality of the diagonals can be verified as follows:
    \begin{align}
        \Delta_{y}[x, x] 
        &= \delta_{f(x), y} && \text{(By definition)} \\
        &= \delta_{\sigma(f(x)), \sigma(y)} && \text{($\sigma$ is a bijection)} \\ 
        &= \delta_{f(\pi(x)), \sigma(y)} && \text{(\cref{eq:symmetry_defn_mu_and_sigma})} \\ 
        &= (P^{-1}_{\pi} \Delta_{\sigma(y)} P_{\pi})[x,x]. 
    \end{align}
\end{proof}

Now we prove that symmetries generate new solutions from a given solution to the Barnum-Saks-Szegedy SDP.
\begin{lem}\label{lem:bss_sdp_symmetry}
    Consider the program $P(f, k, \varepsilon)$, and suppose that $(\pi, \rho, \sigma)$ is a symmetry of $f$. 
    If the matrices $M^{(t)}_i, N^{(t)}, \Gamma_y$ constitute a solution to $P(f, k, \varepsilon)$, then so do matrices 
    \begin{align}
        \widetilde{M}^{(t)}_i &\coloneqq P^{-1} M^{(t)}_{\rho(i)} P & \widetilde{N}^{(t)} &\coloneqq P^{-1} N^{(t)} P & \widetilde{\Gamma}_y &\coloneqq P^{-1} \Gamma_{\sigma(y)} P \label{eq:bss_sdp_new_sln_via_symmetry}
    \end{align}
    where $P \coloneqq P_{\pi}$ is the permutation matrix corresponding to $\pi$. 
\end{lem}
\begin{proof} 
    First note that the initial constraint (\cref{eq:bss-initial-step}) is satisfied since conjugation by a permutation matrix leaves the all-ones matrix invariant:
    \begin{align}
        J = P^{-1} J P = P^{-1} \Biggl(\sum_{i=1}^n M^{(0)}_i + N^{(0)} \Biggr) P = \sum_{i=1}^n \widetilde{M}^{(0)}_i + \widetilde{N}^{(0)}.
    \end{align} 

    Next, \cref{eq:symmetry_defn_E_mxs} implies $E_{i} = P^{-1} E_{\rho(i)} P$. Therefore,
    \begin{align}
        E_i \odot \widetilde{M}^{(t)}_{i} &= \Bigl(P^{-1} E_{\rho(i)} P\Bigr) \odot \Bigl(P^{-1} M_{\rho(i)}^{(t)} P\Bigr) \\
        &= P^{-1} \Bigl(E_{\rho(i)} \odot M^{(t)}_{\rho(i)}\Bigr) P.
    \end{align}

    Consider the forward step (\cref{eq:bss_forward_step}), and notice that the conjugation by $P$ commutes with the sums on both sides of the equation. The permutation $\rho$ leaves the sums over $i$ invariant, while the permutation $\sigma$ leaves the sum over $y$ invariant (for the case of $t = k-1$). Thus the forward step (\cref{eq:bss_forward_step}) is also satisfied by the matrices defined in \cref{eq:bss_sdp_new_sln_via_symmetry}.
    
    Finally, by \cref{prop:conjugating_delta_mx},
    \begin{align}
        \Delta_y \odot \widetilde{\Gamma}_y &= \Bigl(P^{-1} \Delta_{\sigma(y)} P\Bigr) \odot \Bigl(P^{-1} \Gamma_{\sigma(y)} P\Bigr) \\
        &= P^{-1} \Bigl(\Delta_{\sigma(y)} \odot \Gamma_{\sigma(y)} \Bigr) P \\ 
        &\geq (1-\varepsilon) P^{-1} \Delta_{\sigma(y)} P &\text{(By \cref{eq:bss_entrywise_ineq})} \\
        &= (1-\varepsilon) \Delta_y.
    \end{align}
    Therefore the final constraint, \cref{eq:bss_entrywise_ineq}, is also satisfied by the new matrices.   
\end{proof}

By convexity, averaging over a group of symmetries gives another solution of $P(f, k, \varepsilon)$.

\begin{cor}\label{cor:bss_sdp_averaging_symmetries}
    Suppose that the premise of \cref{lem:bss_sdp_symmetry} holds, and let $H$ be a subgroup of the group of symmetries of $f$. For $t \in \{0, \ldots, k-1\}$, $i \in \{0, \ldots, n-1\}$, and $y \in Y$, let
    \begin{align}
        \widetilde{M}^{(t)}_i &\coloneqq \frac{1}{\abs{H}} \sum_{(\pi, \rho, \sigma)} P_{\pi}^{-1} M^{(t)}_{\rho(i)} P_{\pi} & \widetilde{N}^{(t)} &\coloneqq \frac{1}{\abs{H}} \sum_{(\pi, \rho, \sigma)} P_{\pi}^{-1} N^{(t)} P_{\pi} & \widetilde{\Gamma}_y &\coloneqq \frac{1}{\abs{H}}\sum_{(\pi, \rho, \sigma)} P_{\pi}^{-1} \Gamma_{\sigma(y)} P_{\pi}.
    \end{align}
    where the summations range over $H$. Then these matrices satisfy $P(f, k, \varepsilon)$. Furthermore, this solution is invariant under the action of the elements of $H$, in the sense that
    \begin{align}
        &\forall (\pi, \rho, \sigma) \in H: & \widetilde{M}_{i}^{(t)} &= P_{\pi}^{-1} \widetilde{M}_{\rho(i)}^{(t)} P_{\pi} & \widetilde{N}^{(t)} &= P_{\pi}^{-1} \widetilde{N}^{(t)} P_{\pi} & \widetilde{\Gamma}_y &= P_{\pi}^{-1} \widetilde{\Gamma}_{\sigma(y)} P_{\pi}.
    \end{align}
\end{cor}

\subsection{Translation-Invariant Algorithms}\label{subsec:translation_invariant_algorithms}

In this section we prove our main theorem, showing that the class of translation-invariant quantum algorithms (\cref{defn:translation_invariant_alg}) is optimal for the ordered search problem. We begin by studying the symmetries of the problem. 

Recall that the possible inputs are the cyclic translates of the $2n$-bit string ${w \coloneqq 1^n 0^n}$. The set of possible query locations is $\mathbb{Z}/2n$, while the set of possible outputs is $\mathbb{Z}/n$. 
Let ${\tau = (0\;1\;\ldots\;(2n-1))}$ and let $T = P_{\tau}$ denote the permutation matrix corresponding to $\tau$. Then regarding the inputs as column vectors, multiplication by $T$ permutes the set of inputs $X$, which we can express as $X = \{T^{j}w: j \in \mathbb{Z}/2n\}$. The function we need to compute is
\begin{align}
    f(T^{j} w) = j \bmod n.
\end{align}
In particular, we have $f(T x) = f(x) + 1$, where the addition is done modulo $n$. Along these lines we identify a cyclic group of symmetries of $f$ as follows.

\begin{prop}\label{prop:osp_cyclic_symmetry_group}
    The group of symmetries of the ordered search problem contains the cyclic subgroup generated by $(T = P_{\tau}, \tau^{-1}, \mu)$, where $\tau = (0\;1\;\ldots\;(2n-1))$ and $\mu = (0\;1\;\ldots\;(n-1))$.
\end{prop}

Here we abuse notation by taking the first coordinate of a symmetry to be a $2n \times 2n$ matrix $T$. We can identify $T$ with a permutation of the inputs $X$ by regarding the inputs as $2n$-dimensional column vectors. Then multiplication by $T$ permutes the elements of $X$.

\begin{proof}
    First observe that the group generated by $(T, \tau^{-1}, \mu)$ with the group operation being composition (as in \cref{eq:composition_of_symmetries}) is isomorphic to $\mathbb{Z}/2n$. This is because $T$ and $\tau^{-1}$ have order $2n$ and $\mu$ has order $n \mid 2n$.

    It remains to verify that $(T, \tau^{-1}, \mu)$ is indeed a symmetry of $f$. First we verify that \cref{eq:symmetry_defn_E_mxs} holds for $(\pi, \rho, \sigma) = (T, \tau^{-1}, \mu)$:
    \begin{align}
        E_{\tau^{-1}(i)}[x,x'] &= E_{i-1}[x,x'] \\
        &= (-1)^{x_{i-1} + x'_{i-1}} \\ 
        &= (-1)^{(T x)_i + (T x')_i} \\ 
        &= E_i[T x, T x'],
    \end{align}
    where all arithmetic is done modulo $2n$.
    
    As noted above, $f(T x) = f(x) + 1$, and as such $\mu(f(x)) = f(x) + 1 = f(T x)$, with arithmetic done modulo $n$ this time. This is just \cref{eq:symmetry_defn_mu_and_sigma} for $(\pi, \rho, \sigma) = (T, \tau^{-1}, \mu)$.
\end{proof}

Averaging over this group of symmetries, we obtain solutions of a special form.

\begin{cor}\label{cor:osp_bss_symmetrized_solutions} Consider the Barnum-Saks-Szegedy SDP (\cref{defn:bss-sdp}) for the $n$-element symmetrized ordered search problem. The solution matrices $M_{i}^{(t)}$, $N^{(t)}$, and $\Gamma_{y}$ (for $t \in \{0, \ldots, k-1\}$, $i \in \{0, \ldots, 2n-1\}$, and $y \in \{0, \ldots, n-1\}$) may be assumed without loss of generality to satisfy
    \begin{align}
        T^{-1} M^{(t)}_{i} T &= M^{(t)}_{i+1}, & T^{-1} N^{(t)} T &= N^{(t)}, &  T^{-1} \Gamma_y T &= \Gamma_{y-1},
    \end{align}
    where the subscript of $M^{(t)}_{i+1}$ is taken modulo $2n$, and the subscript of $\Gamma_{y-1}$ is taken modulo $n$.
    Furthermore, the constant matrices $E_i$ satisfy $E_{i+1} = T^{-1} E_{i} T$, where the subscript is taken modulo $2n$.
\end{cor}

\begin{proof}
Let $(T, \tau^{-1}, \mu)$ be as in \cref{prop:osp_cyclic_symmetry_group}, and let $H$ denote the generated group of symmetries. Appealing to \cref{cor:bss_sdp_averaging_symmetries} for the group $H$, we have another SDP solution given by the matrices
    \begin{align}
        &\widetilde{M}_i^{(t)} \coloneqq \frac{1}{2n} \sum_{j=0}^{2n-1} T^{-j} M_{i-j}^{(t)} T^{j},
        & &\widetilde{N}^{(t)} \coloneqq \frac{1}{2n} \sum_{j=0}^{2n-1} T^{-j} N^{(t)} T^{j}, 
        & &\widetilde{\Gamma}_y \coloneqq \frac{1}{2n} \sum_{j=0}^{2n-1} T^{-j} \Gamma_{y+j} T^{j}. 
    \end{align}
    Then
    \begin{align}
        T^{-1} \widetilde{M}^{(t)}_i T &= \frac{1}{2n} \sum_{j=0}^{2n-1} T^{-(j+1)} M^{(t)}_{i-j} T^{(j+1)} \\ 
        &= \frac{1}{2n} \sum_{j=0}^{2n-1} T^{-j} M^{(t)}_{i+1-j} T^{j} &&\text{(since $T^{2n} = I$)} \\
        &= \widetilde{M}^{(t)}_{i+1}.
    \end{align}
    The identities $T^{-1} \widetilde{N}^{(t)} T = \widetilde{N}^{(t)}$ and $T^{-1} \Gamma_y T = \Gamma_{y-1}$ can be verified analogously.
    Finally, the \say{furthermore} part follows directly from the definition of symmetries.
\end{proof}

From \cref{prop:no-null-queries}, we can restrict our attention to algorithms without controlled query access. In other words, we can drop the matrices $N^{(t)}$ corresponding to null queries.
    
\begin{prop}\label{prop:bss_osp_simplifies}
The Barnum-Saks-Szegedy SDP characterizing $\varepsilon$-error quantum algorithms for the ordered search problem is equivalent to the following: 
\begin{align}
    \Tr_j\bigl(M^{(0)}\bigr)  &= 2 \label{eq:osp-bss-initial-step} &&(0\leq j < 2n) \\  
    \Tr_j\bigl(E_0 \odot M^{(t)}\bigr) &= \Tr_j\bigl(M^{(t+1)}\bigr) &&(0 \leq j < 2n,\, 0\leq t < k)\label{eq:osp-bss-forward-step} \\ 
    M^{(k)}[0,0], M^{(k)}[n, n] &\geq 1-\varepsilon &&\label{eq:osp-bss-final-step} \\ 
    M^{(t)}  &\in \psd{2n} &&(0\leq t\leq k), \label{eq:osp-bss-semidefinite}
\end{align}
where
\begin{align}
    E_0 \coloneqq \begin{pmatrix}
    J_n & -J_n \\
    -J_n & J_n
  \end{pmatrix}. \label{eq:e1_matrix_for_bss_osp}
\end{align}
\end{prop}

The constraint $\Tr_j\bigl(M^{(0)}\bigr) = 2$ may seem somewhat unnatural. This is a result of the symmetrization of the problem (see the discussion preceding \cref{prob:osp-symmetrized}), by which the size of the input strings is doubled. In \cref{prop:bss-sdp-abba-form,prop:bss-sdp-a-minus-a-form} we show that the above program is equivalent to an analogous one where the matrices $M^{(t)} \in \mathbf{S}_{+}^{2n}$ are replaced by matrices $A^{(t)} \in \mathbf{S}_{+}^{n}$ of half the size, with the initial constraint becoming $\Tr\bigl(A^{(0)}\bigr) = 1$.

\begin{proof} 
    Let $X$ be the set of all $2n$ cyclic translates of the bitstring $w \coloneqq 1^n 0^n$, and let $Y = \mathbb{Z}/n$. Recall that the elements of $X$ can be written as $T^{j} w$ where $T = P_{\tau}$ for $\tau = (0\;1\;\ldots\;(2n-1))$ and $j \in \mathbb{Z}/2n$.
    
    By \cref{prop:no-null-queries} the null query matrices $N^{(t)}$ can be removed from the Barnum-Saks-Szegedy SDP (\cref{defn:bss-sdp}) without loss of generality. We are left with
    \begin{align}
        J_{2n} &= \sum_{i=0}^{2n-1} M_{i}^{(0)} \label{eq:bss-sdp-nnq-initial-step} \\ 
        \sum_{i=0}^{2n-1} E_i \odot M_i^{(t)} &= \begin{cases}
        \sum_{i=0}^{2n-1} M_i^{(t+1)}, &t = 0, \ldots, k-2 \\ 
        \sum_{y=0}^{n-1} \Gamma_y, &t = k-1
    \end{cases} \label{eq:bss-sdp-nnq-forward-step} \\
        \Delta_y \odot \Gamma_y &\geq (1-\varepsilon) \Delta_y\quad  (0 \leq y < n) \label[ineq]{eq:bss-sdp-nnq-entrywise-ineq} \\
        M_i^{(t)}, \Gamma_y &\in \psd{2n}.
    \end{align}
    
    By \cref{cor:osp_bss_symmetrized_solutions} we can assume that the matrices $M^{(t)}_i$, $\Gamma_y$, and $E_i$ satisfy $M^{(t)}_{i+1} = T^{-1} M^{(t)}_{i} T$, $T^{-1} \Gamma_y T  = \Gamma_{y-1}$, and $E_{i+1} = T^{-1} E_i T$, respectively. In particular, we have
    \begin{align}
        M^{(t)}_{i} &= T^{-i} M^{(t)}_0 T^{i} \\ 
        E_{i} &= T^{-i} E^{(t)}_0 T^{i}, 
    \end{align}
    and in turn,
    \begin{align}
        E_{i} \odot M^{(t)}_{i} &= \bigl(T^{-i} E_0 T^{i}\bigr) \odot \bigl(T^{-i} M^{(t)}_0 T^{i}\bigr) \\ 
        &= T^{-i} \bigl(E_0 \odot M^{(t)}_0\bigr) T^{i}.
    \end{align}
    Hence the matrices $\sum_{i=0}^{2n-1} M^{(t)}_i$ and $\sum_{i=0}^{2n-1} E_i \odot M^{(t)}_i$ for $0 \leq t < k$ are Toeplitz, in addition to being symmetric. More specifically, all of the entries of $\sum_{i=0}^{2n-1} M^{(t)}_i$ on the $j$th superdiagonal equal $\Tr_j(M^{(t)}_0)$, while all the entries of $\sum_{i=0}^{2n-1} E_i \odot M^{(t)}_i$ on the $j$th superdiagonal equal $\Tr_j\bigl(E_0 \odot M^{(t)}_0\bigr)$:
    \begin{align}
        \sum_{i=0}^{2n-1} M^{(t)}_i &= \toep\Bigl(\Tr_0\bigl(M^{(t)}_0\bigr), \ldots, \Tr_{2n-1}\bigl(M^{(t)}_0\bigr)\Bigr) \\ 
        \sum_{i=0}^{2n-1} E_i \odot M^{(t)}_i &= \toep\Bigl(\Tr_0\bigl(E_0 \odot M^{(t)}_0\bigr), \ldots, \Tr_{2n-1}\bigl(E_0 \odot M^{(t)}_0\bigr)\Bigr).
    \end{align}
    It follows by \cref{eq:bss-sdp-nnq-forward-step} that $\sum_{y=0}^{n-1} \Gamma_y$ is Toeplitz as well. Since $T^{-1} \Gamma_y T = \Gamma_{y-1}$, we have
    \begin{align}
        \sum_{y=0}^{n-1} \Gamma_y 
        &= \sum_{y=0}^{n-1} T^{y} \Gamma_0 T^{-y} \\
        &= \frac{1}{2} \toep\Bigl(\Tr_0(\Gamma_0), \Tr_1(\Gamma_0), \ldots, \Tr_{2n-1}(\Gamma_0)\Bigr).
    \end{align}
    Therefore the constraints in \cref{eq:bss-sdp-nnq-initial-step,eq:bss-sdp-nnq-forward-step} can be replaced by the following:
    \begin{align}
        1 &= \Tr_j\bigl(M^{(0)}_0\bigr) &&(0 \leq j < 2n) \\ 
        \Tr_j\bigl(E_0 \odot M^{(t)}_0\bigr) &= \begin{cases}
            \Tr_j\bigl(M^{(t+1)}_0\bigr) &t = 0, \ldots, k-2 \\ 
            \frac{1}{2} \Tr_j(\Gamma_0), & t = k-1.
        \end{cases} &&(0 \leq j < 2n)
    \end{align}
    Thus we eliminated the matrices $M^{(t)}_i, E_i$ for $i > 0$. To eliminate the matrices $\Gamma_y$ (for $y > 0$) as well, we express the constraint
    \begin{align}
        \Delta_y \odot \Gamma_y &\geq (1-\varepsilon) \Delta_y \quad (0 \leq y < n)
    \end{align}
    in terms of $\Gamma_0$. To this end, by definition we have $\Delta_y = T^{-y} \Delta_0 T^{y}$, and accordingly, 
    \begin{align}
        \Delta_y \odot \Gamma_y = (T^{-y} \Delta_0 T^{y}) \odot (T^{-y} \Gamma_0 T^{y}) = T^{-y} (\Delta_0 \odot \Gamma_0) T^{y}.
    \end{align}
    Thus $\Delta_y \odot \Gamma_y \geq (1-\varepsilon) \Delta_y$ holds for all $y \in \{0, \ldots, n-1\}$ if and only if 
    \begin{align}
        \Delta_0 \odot \Gamma_0 &\geq (1-\varepsilon) \Delta_0.
    \end{align}

    In conclusion, we have the following SDP:
    \begin{align}
        1 &= \Tr_j\bigl(M^{(0)}_0\bigr) &&(0 \leq j < 2n) \label{eq:osp-sdp-end-of-proof-initial-constraint} \\ 
        \Tr_j\bigl(E_0 \odot M^{(t)}_0\bigr) &= \begin{cases}
            \Tr_j\bigl(M^{(t+1)}_0\bigr) &t = 0, \ldots, k-2 \\ 
            \frac{1}{2} \Tr_j(\Gamma_0), & t = k-1
        \end{cases} &&(0 \leq j < 2n) \label{eq:osp-sdp-end-of-proof-forward-step} \\ 
        \Gamma_0[0,0], \Gamma_0[n,n] &\geq 1 - \varepsilon \\ 
        M_0^{(t)}, \ldots, M_0^{(k-1)}, \Gamma_0 &\in \psd{2n}.
    \end{align}
    This is equivalent to the claimed SDP, as can be seen by letting $M^{(t)} \coloneqq 2 M_0^{(t)}$ for $0 \leq t < k$, and $M^{(k)} \coloneqq \Gamma_0$.
\end{proof}

In the next two propositions we show that the solution matrices $M^{(t)}$ can, without loss of generality, be assumed to be of the form
\begin{align}
    M^{(t)} = \begin{pmatrix}
            A^{(t)} & (-1)^t A^{(t)} \\ 
            (-1)^t A^{(t)} & A^{(t)}
        \end{pmatrix}. \label{eq:osp-bss-matrix-form}
\end{align}

\begin{prop}\label{prop:bss-sdp-abba-form}
    Without loss of generality, the matrices $M^{(t)}$ in the above SDP (\cref{eq:osp-bss-initial-step,eq:osp-bss-forward-step,eq:osp-bss-final-step,eq:osp-bss-semidefinite}) may be assumed to have the form 
    \begin{align}
        M^{(t)} = \begin{pmatrix}
            U^{(t)} & V^{(t)} \\ 
            V^{(t)} & U^{(t)}
      \end{pmatrix}, \label{eq:osp_bss_abba_form}
    \end{align}
    where $U^{(t)} \in \psd{n}$ and $V^{(t)}$ is a symmetric $n \times n$ matrix. Moreover, for any such solution the main diagonals of $U^{(t)}$ and $(-1)^{t} V^{(t)}$ are identical.
\end{prop}
\begin{proof}
We first establish the \say{moreover} part. Assume that the matrices $(M^{(t)})_{t=0}^{k}$ are a feasible solution of the prescribed form. Then by \cref{eq:osp-bss-initial-step} we have $1 = \tr(U^{(0)}) = \tr(V^{(0)})$. \Cref{eq:osp-bss-forward-step} implies that $1 = \tr(U^{(t)}) = (-1)^t \tr(V^{(t)})$ for each $t \in \{0, \ldots, k\}$. Since $M^{(t)}$ is positive semidefinite, we must have 
\begin{align}
    -U^{(t)}[j,j] &\leq V^{(t)}[j,j] \leq U^{(t)}[j,j] \quad\quad (\forall j \in \{0, \ldots, n-1\}). \label{eq:mx_a_diagonal_sandwiches_mx_b}
\end{align}
When $t$ is odd, by the identity $\tr(U^{(t)}) = (-1)^{t} \tr(V^{(t)})$ the first inequality must be saturated. Similarly, when $t$ is even the second inequality must be saturated. In other words, $U^{(t)}[j,j] = (-1)^{t} V^{(t)} [j,j]$ as claimed. 

We now prove the rest of the claim. Given an arbitrary solution $(\widehat{M}^{(t)})_{t=0}^{k}$ we construct another solution $(M^{(t)})_{t=0}^k$ that consists of matrices of the stated form. Suppose that the matrices $\widehat{M}^{(t)}$ are a feasible solution, and for each $t \in \{0, \ldots, k\}$ write 
\begin{align}
    \widehat{M}^{(t)} = \begin{pmatrix}
        A^{(t)} & B^{(t)} \\ 
        C^{(t)} & D^{(t)}
  \end{pmatrix}.
\end{align}
Now define
\begin{align}
    \widecheck{M}^{(t)} \coloneqq \begin{pmatrix}
    D^{(t)} & C^{(t)} \\
    B^{(t)} & A^{(t)}
    \end{pmatrix}
\end{align}
where we swap the rows and swap the columns of $\widehat{M}^{(t)}$. Clearly $\widecheck{M}^{(t)} \succeq 0$ if and only if $\widehat{M}^{(t)} \succeq 0$. We claim that the set of matrices $\widecheck{M}^{(t)}$ is another feasible solution. 
Given $j \in \{0, \ldots, 2n-1\}$, define $j'$ by letting $j' \coloneqq j$ if $j\leq n$, and $j' \coloneqq j - 2n$ otherwise. Then, 
\begin{align}
    \Tr_j\bigl(\widecheck{M}^{(t)}\bigr) &= \tr_{j'}\bigl(D^{(t)}\bigr) + \tr_{j-n}\bigl(C^{(t)}\bigr) + \tr_{j'}\bigl(A^{(t)}\bigr) + \tr_{j-n}\bigl(B^{(t)}\bigr) = \Tr_j\bigl(\widehat{M}^{(t)}\bigr).
\end{align}
Similarly it is easy to see that $\Tr_j(E_0 \odot \widecheck{M}^{(t)}) = \Tr_j(E_0 \odot \widehat{M}^{(t)})$ for $j \in \{0, \ldots, 2n-1\}$. 
Thus, \cref{eq:osp-bss-initial-step,eq:osp-bss-forward-step} are satisfied by the matrices $\widecheck{M}^{(t)}$. \Cref{eq:osp-bss-final-step} is also satisfied, since 
\begin{align}
    \widecheck{M}^{(k)}[0,0] = \widehat{M}^{(k)}[n,n] \quad \text{ and } \quad \widecheck{M}^{(k)}[n,n] = \widehat{M}^{(k)}[0,0].
\end{align}

We conclude that the matrices $\widecheck{M}^{(t)}$ constitute another feasible solution. Then, by convexity so do the matrices $M^{(t)} \coloneqq (\widehat{M}^{(t)} + \widecheck{M}^{(t)})/2$, and these are of the desired form.
\end{proof}

\begin{prop}\label{prop:bss-sdp-a-minus-a-form}
    Suppose that the SDP in \cref{eq:osp-bss-initial-step,eq:osp-bss-forward-step,eq:osp-bss-final-step,eq:osp-bss-semidefinite} has a solution  $(\widehat{M}^{(t)})_{t=0}^k$ such that the matrices $\widehat{M}^{(t)}$ are of the form
        \begin{align}
            \widehat{M}^{(t)} = \begin{pmatrix}
                U^{(t)} & V^{(t)} \\
                V^{(t)} & U^{(t)}
            \end{pmatrix}.
        \end{align}
    Then the matrices ${M}^{(t)} \coloneqq \frac{1}{2} \bigl(\widehat{M}^{(t)} + (-1)^t \widecheck{M}^{(t)}\bigr)$ give another solution where
    \begin{align}
        \widecheck{M}^{(t)} \coloneqq \begin{pmatrix}
            V^{(t)} & U^{(t)} \\ 
            U^{(t)} & V^{(t)}
        \end{pmatrix}.
    \end{align}
\end{prop}

\begin{proof}
    First we show that the matrices $M^{(t)}$ are positive semidefinite. To see this, note that since $\widehat{M}^{(t)}$ is positive semidefinite, so is
    \begin{align}
        \frac{1}{2} \begin{pmatrix}I_n & I_n \\ I_n & -I_n\end{pmatrix}^{\dagger} \widehat{M}^{(t)} \begin{pmatrix}I_n & I_n \\ I_n & -I_n\end{pmatrix} &= \begin{pmatrix}
            U^{(t)} + V^{(t)} & 0  \\ 0 & U^{(t)} - V^{(t)} \end{pmatrix} \succeq 0.
    \end{align}
    In particular, the matrices $U^{(t)} + V^{(t)}$ and $U^{(t)} - V^{(t)}$ are both positive semidefinite. Thus so is $M^{(t)}$, being the Kronecker product of two positive semidefinite matrices:
    \begin{align}
        M^{(t)} = (U^{(t)} + (-1)^t V^{(t)}) \otimes \ \begin{pmatrix} 1 & (-1)^t \\ (-1)^t & 1 \end{pmatrix} \succeq 0.
    \end{align}
    
    It is easy to verify that $\Tr_j\bigl(\widecheck{M}^{(t)}\bigr) = \Tr_{j-n}\bigl(\widehat{M}^{(t)}\bigr)$ for $j \in \{0, \ldots, 2n-1\}$. Analogously, ${\Tr_j\bigl(E_0 \odot \widecheck{M}^{(t)}\bigr) = -\Tr_{j-n}\bigl(\widehat{M}^{(t)}\bigr)}$. Then, by linearity
    \begin{align}
        \Tr_j\bigl(E_0 \odot M^{(t)}\bigr) &= \frac{1}{2}\Bigl(\Tr_j\bigl(E_0 \odot \widehat{M}^{(t)}\bigr) - (-1)^t \Tr_{j-n}\bigl(E_0 \odot \widehat{M}^{(t)}\bigr)\Bigr) \\
        &= \frac{1}{2}\Bigl(\Tr_j\bigl(\widehat{M}^{(t+1)}\bigr) + (-1)^{t+1} \Tr_{j-n}\bigl(\widehat{M}^{(t+1)}\bigr)\Bigr) &(\text{By \cref{eq:osp-bss-forward-step}})\\ 
        &= \frac{1}{2}\Bigl(\Tr_j\bigl(\widehat{M}^{(t+1)}\bigr) + (-1)^{t+1} \Tr_{j}\bigl(\widecheck{M}^{(t+1)}\bigr)\Bigr) \\ 
        &= \Tr_j\bigl(M^{(t+1)}\bigr).
    \end{align}
    This shows that \cref{eq:osp-bss-forward-step} is satisfied by the matrices $M^{(t)}$. The initial constraint (\cref{eq:osp-bss-initial-step}) is also satisfied because $\Tr_j(\widehat{M}^{(0)}) = 2$ for all $j \in \{0, \ldots, 2n-1\}$ by assumption, and therefore
    \begin{align}
        \Tr_j\bigl(M^{(0)}\bigr) &= \frac{1}{2} \Bigl(\Tr_j(\widehat{M}^{(0)}) + \Tr_{j}(\widecheck{M}^{(0)})\Bigr) = \frac{1}{2} \Bigl(\Tr_j(\widehat{M}^{(0)}) + \Tr_{j-n}(\widehat{M}^{(0)})\Bigr) = 2.
    \end{align}
    Finally, by \cref{prop:bss-sdp-abba-form} the main diagonal of $U^{(k)}$ is identical to the main diagonal of $(-1)^k V^{(k)}$, and consequently
    \begin{align}
        M^{(k)}[0,0] &= \frac{1}{2}\Bigl(U^{(k)}[0,0] + (-1)^k V^{(k)}[0,0]\Bigr) \\ 
        &= \frac{1}{2} \Bigl(U^{(k)}[0,0] + (-1)^{2k} U^{(k)}[0,0]\Bigr) \\ 
        &= U^{(k)}[0,0] \geq 1 - \varepsilon.
    \end{align}
    Analogously $M^{(k)}[n,n] \geq 1-\varepsilon$, and the final constraint (\cref{eq:osp-bss-final-step}) is satisfied as well.
\end{proof}

We can now present and prove our main result, which asserts that the class of translation-invariant quantum algorithms with no workspace (\cref{defn:translation_invariant_alg}) is optimal for the ordered search problem. We show this by reducing the Barnum-Saks-Szegedy SDP to our semidefinite programming characterization of translation-invariant algorithms for ordered search, which we restate below. 

\thmCLP*

We prove our main theorem by showing that for the ordered search problem the Barnum-Saks-Szegedy SDP is equivalent to the above semidefinite program.

\thmMainResult*

\begin{proof} By \cref{prop:bss-sdp-a-minus-a-form}, we can assume without loss of generality that the matrices $M^{(t)}$ in the SDP of \cref{prop:bss_osp_simplifies} characterizing $\varepsilon$-error algorithms for the ordered search problem are of the form 
\begin{align}
    M^{(t)} = \begin{pmatrix}
    A^{(t)} & (-1)^t A^{(t)} \\ 
    (-1)^t A^{(t)} & A^{(t)}
    \end{pmatrix}.
\end{align}
Observe that $M^{(t)}$ is positive semidefinite if and only if $A^{(t)}$ is (in particular, $M^{(t)}$ is the tensor product of $A^{(t)}$ with a positive semidefinite matrix). Rewriting the constraints in terms of the submatrices $A^{(t)}$, we obtain the equivalent program
\begin{align}
    \tr_j{A^{(0)}} + \tr_{n-j}{A^{(0)}} &= 1 &(0 \leq j < n) \label{eq:osp-bss-simplified-initial-step}\\ 
    \tr_j{A^{(t)}} + (-1)^{t+1} \tr_{j-n}{A^{(t)}} &= \tr_j{A^{(t+1)}} + (-1)^{t+1} \tr_{j-n}{A^{(t+1)}} &(0 \leq j < n, 0\leq t < k) \label{eq:osp-bss-simplified-intermediate-steps}\\ 
    A^{(k)}[0,0] &\geq 1-\varepsilon \label{eq:osp-bss-simplified-final-step}\\
    \tr{A^{(t)}} &= 1 &(0\leq t \leq k) \label{eq:osp-bss-simplified-normalization} \\
    A^{(t)} &\in \psd{n} &(0\leq t\leq k).
\end{align}
In a proof deferred to the Appendix (\cref{lem:osp_bss_simplified_first_constraint_unique_sln}), we show that \cref{eq:osp-bss-simplified-initial-step} has a unique positive semidefinite solution, namely $A^{(0)} = \frac{1}{n} J_n$. 
Replacing \cref{eq:osp-bss-simplified-initial-step} with
this constant choice of $A^{(0)}$,
we end up with the SDP stated in \cref{thm:clp}. Thus we have transformed the SDP characterizing all $\varepsilon$-error quantum algorithms for the ordered search problem into the SDP of \cref{thm:clp} characterizing translation-invariant algorithms for this problem.
\end{proof}

As a consequence of this result, we see that the feasibility of $\varepsilon$-error $k$-query translation-invariant ordered search of an $n$-element list is monotonic in $n$.

\begin{cor}\label{cor:monotonicity}
If there is an $\varepsilon$-error $k$-query translation-invariant quantum algorithm for searching a list of length $n$, then there is also an $\varepsilon$-error $k$-query translation-invariant algorithm for a list of any length less than $n$.
\end{cor}
\begin{proof}
In general, algorithms for ordered search clearly satisfy this monotonicity, since an algorithm for a larger instance of \cref{prob:OSP} (which is equivalent to \cref{prob:osp-symmetrized} as discussed in \cref{sec:osp}) can be applied to a smaller one by padding the input on the right with $1$s. Since feasibility of translation-invariant algorithms is equivalent to feasibility of general algorithms, this property also holds for translation-invariant algorithms.
\end{proof}

Note that without establishing equivalence to general algorithms, this property is not obvious, since padding the input with $1$s does not respect translation invariance. Indeed, for this reason Ref.~\cite{clp} was unable to definitively identify the largest list that could be searched exactly with a given number of queries. This result shows, for example, that the infeasibility of an exact translation-invariant algorithm searching a 606-element list with 4 queries (as shown in Ref.~\cite{clp}) implies the infeasibility of exactly searching an $n$-element list with 4 queries for any $n \ge 606$.

Another corollary of our result is that the algorithm of \cite{bh07} can be made workspace-free. The implementation described by Ben-Or and Hassidim stores a list of $O(n)$ weights, one for each query index, whereas our result implies that there exists an algorithm with no workspace that achieves the same query complexity.

\section{Finding Exact Translation-Invariant Algorithms with Linear Programming}

In this section, we present a linear programming approach to identifying exact translation-invariant quantum algorithms for ordered search. Prior work in this direction utilized nonconvex optimization \cite{fggs99}, gradient descent \cite{bjl04}, and semidefinite programming \cite{clp}. In particular, Ref.~\cite{clp} used a semidefinite programming approach to exhibit an exact, $4$-query translation-invariant algorithm for searching a $605$-element list. The primary downside of such approaches is that the number of scalar variables an SDP solver needs to allocate to represent the matrices scales quadratically with the instance size $n$. In turn, $n$ scales exponentially as a function of the number of queries $k$. This means that memory requirements quickly become prohibitive. Indeed, despite 18 years of improvements to computers since Ref.~\cite{clp}, understanding the best $5$-query algorithm still appears to be out of reach of SDP solvers.

While the exponential scaling of $n$ with $k$ is inevitable, we can improve the memory requirements of representing the program in \cref{thm:fggs-pos-polys} by considering a linear programming relaxation instead. This way, the memory required to represent the problem scales linearly with $n$, a quadratic improvement that allows us to go one step further than Ref.~\cite{clp}. Using this LP characterization, we show that the largest list searchable by an exact translation-invariant algorithm using $k = 5$ queries is of size $n = 7265$. By \cref{thm:main_result}, this is the best one can do with any (not necessarily translation-invariant) $5$-query algorithm.

We do this by exhibiting an explicit solution of the program in \cref{thm:fggs-pos-polys} for $n = 7265$, as well as an LP relaxation of this program that certifies infeasibility for $n = 7266$. By \cref{cor:monotonicity} we conclude that 7265 is the length of the longest list exactly searchable by a 5-query algorithm.

\subsection{Eliminating Half of the Variables}\label{sec:lp-eliminating-variables}

Recall the polynomial characterization of \cite{fggs99} of exact translation-invariant algorithms:
\thmFGGSPosPolys*

Below we show that about half of the polynomial variables $q_t$ in the above program are redundant. This observation allows for a more compact representation of the linear program (\cref{defn:lp}) presented in the next section, leading to an improved runtime performance of our implementation.

\begin{prop}\label{prop:redundant-eqs}
    For any $t \in [k-1]$ the constraint in \cref{eq:fggs-pos-polys-forward-step} involving $q_t$ determines $q_t$ in terms of $q_{t-1}$ and $q_{t+1}$:
    \begin{align}
        a^{(t)} &= \frac{1}{2} \bigl(a^{(t-1)} + a^{(t+1)}\bigr) + \frac{(-1)^t}{2} V_n \bigl(a^{(t-1)} - a^{(t+1)}\bigr) \label{eq:even-odd-trick}
    \end{align}
    where $V_n$ is the $n \times n$ permutation matrix
    \begin{align}
        V_n &= \begin{pmatrix}
            1 & 0 & \cdots & 0 & 0\\
            0 & 0 & \cdots & 0 & 1 \\ 
            \vdots & \vdots & \iddots & \iddots & 0 \\ 
            0 & 0 & 1 & \iddots & \vdots \\
            0 & 1 & 0 & \cdots & 0 
        \end{pmatrix}.
    \end{align}
\end{prop}
\begin{proof}
Recall that 
\begin{align}
q_t(z) &= \frac{1}{2} \sum_{j=0}^{n-1} a^{(t)}_j (z^{j} + z^{-j}).
\end{align}
In \cref{eq:fggs-pos-polys-forward-step}, we evaluate $q_t$ at roots of either $z^{n} - 1$ or $z^{n} + 1$ depending on the parity of $t$. In other words, we want to evaluate $q_t$ at powers of some primitive $2n$th root of unity $\omega$. This can be expressed as a discrete Fourier transform:
\begin{equation}
    \label{eq:powers_by_dft}
    \begin{pmatrix}
        q_t(\omega^0) \\
        \vdots \\
        q_t(\omega^{2n-1}) 
    \end{pmatrix}
    = \tfrac{1}{2} F_{2n} 
    \begin{pmatrix}
        2a_0^{(t)} &
        a_1^{(t)} &
        \cdots &
        a_{n-1}^{(t)} &
        0 &
        a_{n-1}^{(t)} &
        \cdots &
        a_1^{(t)}
    \end{pmatrix}^{\top},
\end{equation}
where $F_{2n}$ is the $2n \times 2n$ discrete Fourier transform matrix.

Now let $\Xi := \text{diag}(1,-1,1,-1,\ldots,-1)$ be the $2n \times 2n$ matrix alternating $1$ and $-1$ along the diagonal. When applied to a vector of evaluations of a polynomial at $\omega^0, \ldots, \omega^{2n-1}$, the matrix $\tfrac{1}{2}(I + \Xi)$ zeros out evaluations at roots of $z^{n} = -1$, while $\tfrac{1}{2}(I - \Xi)$ zeros out evaluations at roots of $z^{n} = +1$. It follows that, overloading notation by writing $q_t \coloneqq \begin{pmatrix} q_t(\omega^0) & \ldots & q_t(\omega^{2n-1}) \end{pmatrix}^{\top}$, \cref{eq:fggs-pos-polys-forward-step} can be written
\begin{align}
    (I + (-1)^{t}\Xi)(q_t - q_{t-1}) = 0. \label{eq:alt-cons-xi}
\end{align}
Observe that $\Xi$ applies a phase $\omega^{nx}$ to the $x$th standard basis vector. Equivalently, $\chi \coloneqq F_{2n}^{-1} \Xi F_{2n}$ acts by swapping the first and last $n$ coordinates of a vector. Applying the inverse Fourier transform to both sides of \cref{eq:alt-cons-xi}, we have 
\begin{align}
    0 &= F_{2n}^{-1} (I + (-1)^{t}\Xi)(q_t - q_{t-1}) \\
    &= (I + (-1)^t \chi) F_{2n}^{-1} (q_t - q_{t-1}).
\end{align}
Exchanging the two halves of the vector on the right-hand side of \Cref{eq:powers_by_dft} has the effect of reversing $a_1^{(t)}, \ldots, a_{n-1}^{(t)}$, and exchanging $2a_0^{(t)}$ with $0$. In other words, \cref{eq:alt-cons-xi} is equivalent to
\begin{align}
    (I + (-1)^{t} V_n)(a^{(t)} - a^{(t-1)}) &= 0. \label{eq:forward-step-matrix-form}
\end{align}
Taking the difference of \cref{eq:forward-step-matrix-form} at $t$ and $t+1$ yields
\begin{align}
    0 &= (I + (-1)^{t} V_n)(a^{(t)} - a^{(t-1)}) - (I + (-1)^{t+1}V_n)(a^{(t+1)} - a^{(t)}),
\end{align}
so
\begin{align}
    a^{(t)} &= \tfrac{1}{2}(a^{(t-1)} + a^{(t+1)}) + \tfrac{1}{2}(-1)^{t} V_n (a^{(t-1)} - a^{(t+1)})
\end{align}
as claimed.
\end{proof}

\subsection{A Linear Programming Relaxation}

We construct a linear program that mirrors \cref{thm:fggs-pos-polys}, whose variables are the polynomial coefficients $a_j^{(t)}$ for $t = 0, 1, \ldots, k$ and $j = -(n-1), \ldots, n-1$. Except for polynomial nonnegativity, all the constraints in \cref{thm:fggs-pos-polys} are linear. Indeed, \cref{eq:fggs-pos-polys-forward-step} is linear since an evaluation of a polynomial is a linear function of the polynomial coefficients. \cref{eq:fggs-pos-polys-norm} is also linear, as it is equivalent to the constraint $a^{(t)}_0 = 1$ (for each $t$). 

To handle the nonnegativity constraint, we consider the following relaxation. Instead of requiring the polynomials to be nonnegative on their entire domain, we require that each polynomial exceeds a threshold $\beta$ on some large but finite set of points $G$. We then let $\beta$ be the objective of the LP to be maximized. Note that this procedure may not always result in positive polynomials. However, once a solution is found, one can certify the feasibility of the original program (\cref{thm:fggs-pos-polys}) as follows. If the objective function is negative then the instance is unsolvable, since there are no polynomials that are all positive even at the points in $G$. On the other hand, if the produced polynomials are positive (which must be checked separately), then the instance is feasible. 

We expand on this methodology below. First, we formally define and prove soundness of our construction.

\begin{defn}\label{defn:lp}
    Let $G$ be a non-empty set of finitely many points on the unit circle. Let $\mathcal{L}(G)$ denote the following linear program, where the maximization is over symmetric real-valued Laurent polynomials $q_t(z) = \frac{1}{2} \sum_{j=0}^{n-1} a^{(t)}_j (z^{j} + z^{-j})$ for $1 \leq t < k$:
    \begin{alignat}{2}
        \beta^{*} = \max \quad \beta& & \label{eq:lp-objective} \\
        \text{s.t.} \quad  q_0 &\equiv F_n & \label{eq:lp-initial} \\
        q_t(z) &= q_{t-1}(z)\quad\quad &(\forall z \text{ with } z^{n} = (-1)^t, \forall t \in [k]) \label{eq:lp-forward-step} \\
        q_k &\equiv 1 & \label{eq:lp-final} \\
        a_0^{(t)} &= 2 &(\forall t = 0, \ldots, k) \label{eq:lp-constant-term} \\ 
        q_t(z) &\geq \beta &(\forall z \in G, \forall t = 0, \ldots, k). \label{eq:lp-non-negativity}
    \end{alignat}
\end{defn}

Our first observation concerns the feasibility of $\mathcal{L}(G)$.

\begin{cor}
    For any finite subset $G$ of the unit circle, the program $\mathcal{L}(G)$ is feasible. 
\end{cor}
\begin{proof}
    Suppose $k$ is even. Consider any assignment of the polynomials $q_2, q_4, \ldots, q_{k-2}$ such that the respective constraints in \cref{eq:lp-constant-term} are satisfied. By \cref{prop:redundant-eqs} the polynomials $q_1, q_3, \ldots, q_{k-1}$ are determined through \cref{eq:even-odd-trick} and the equality constraints in \cref{eq:lp-forward-step} are all satisfied. Let $\beta$ be the minimal evaluation of any one of the polynomials $q_0, \ldots, q_k$ on the set $G$, which is finite since a Laurent polynomial can only have a pole at the origin. This gives a feasible solution.
    
    The case of odd $k$ requires more care. Suppose $k$ is odd, and consider the constraint in \cref{eq:lp-forward-step} for $t = 1$. By the proof of \cref{prop:redundant-eqs} this equation is equivalent to
    \begin{align}
        (I - V_n)(a^{(1)} - a^{(0)}) &= 0.
    \end{align}
    Since the matrix $I - V_n$ is singular, this equation has infinitely many solutions for $a^{(1)}$. Since the above equation imposes no constraints on the coefficient $a^{(1)}_0$, we can fix a solution $a^{(1)}$ satisfying $a^{(1)}_0 = 2$. Next, assign some values to the coefficients of the polynomials $q_3, q_5, \ldots, q_{k-2}$ such that the respective constraints in \cref{eq:lp-constant-term} are satisfied. By \cref{prop:redundant-eqs} our choices determine the remaining polynomials ($q_2, q_4, \ldots, q_{k-1}$), and the equality constraints in \cref{eq:lp-forward-step} are all satisfied. As before, $\beta$ can be chosen to be the minimal value of the polynomials on the set $G$.
\end{proof}

The linear program $\mathcal{L}(G)$ provides an alternative way of certifying that there is no exact $k$-query translation-invariant algorithm for searching an $n$-element list. Indeed, if $\mathcal{L}(G)$ has a negative optimal objective value $\beta^{*} < 0$ then there do not exist nonnegative polynomials $q_0, \ldots, q_k$ satisfying the conditions of \cref{thm:fggs-pos-polys}. The converse, however, is not true: $\mathcal{L}(G)$ having a positive objective value $\beta^{*} > 0$ does not imply that there exists an exact $k$-query translation-invariant algorithm for searching an $n$-element list, because the polynomial could be negative at points outside $G$. 

We can overcome this limitation as follows. Suppose that for some choice of $G$, a solution consisting of the polynomials $q_0, \ldots, q_k$ attains a positive objective value $\beta > 0$. In the Chebyshev representation, the polynomials $q_t$ can be represented as functions from $\mathbb R \to \mathbb R$ (see for example \cite[Section 8.3]{trefethen96}). In this representation, the interval $[-1, 1]$ corresponds to the inputs on the unit circle. For each $q_t$ in the Chebyshev basis, by considering the derivative $q_t'$ we can check each of the local extrema of $q_t(x)$ within the interval $[-1, 1]$ for positivity, i.e., we can verify that $q_t'(x) = 0 \implies q_t(x) > 0$. If we do the same on the boundaries, this certifies that the polynomial is nonnegative on the unit circle. In this way, any solution returned by the linear program $\mathcal{L}(G)$ can be checked for positivity.

\subsection{An Iterative Approach}\label{sec:iterative_approach}

As discussed above, the program $\mathcal{L}(G)$ can be used to certify that there is no exact, $k$-query translation-invariant quantum algorithm for searching an $n$-element list. However, this will only succeed for an appropriate choice of the set $G$. To this end, we use an iterative algorithm that heuristically constructs a sequence of increasing subsets $G_0 \subset G_1 \subset \cdots$ by solving the the corresponding sequence $\mathcal{L}(G_0), \mathcal{L}(G_1), \ldots$ of linear programs. In this process, the set of constraints keeps increasing, and the feasible sets keep shrinking. While the process is not guaranteed to terminate, in practice we find that it does after a reasonable number of iterations. \Cref{alg:iteratedlp} gives pseudocode for this procedure.

\begin{algorithm}[!htb]
\caption{Finding exact translation-invariant algorithms} \label{alg:iteratedlp}
\begin{algorithmic}[1]
\State Let $G_0 \subset \{z \in \mathbb{C} : |z|=1\}$ be such that $\mathcal{L}(G_0)$ is bounded
\For{$i = 0, 1, \dots$}
    \State Let $\beta^{*}_i$ be the optimal value of $\mathcal{L}(G_i)$ and $q_0, \ldots, q_k$ be polynomials attaining that value
    \If{$\beta^{*}_i < 0$} \Return $G_i$ \EndIf \label{ln:return-grid}
    \State Let $M$ be the set of minima\footnotemark{} of $q_0, \ldots, q_k$  that are negative
    \If{$M = \emptyset$} \label{ln:if-m-empty}
        \State \Return $q_0, \ldots, q_k$  \label{ln:return-polys}
    \Else 
        \State Let $G_{l+1} \coloneqq G_l \cup M$
    \EndIf
\EndFor
\end{algorithmic}
\end{algorithm}
\footnotetext{More precisely, let $M$ be the set of complex numbers on the unit circle corresponding to the minima of the Chebyshev representations of $q_0, \ldots, q_k$ restricted to $[-1,1]$ that are negative, as well as any point corresponding to the boundary of $[-1,1]$ at which a Chebyshev polynomial is negative.}

\begin{prop}
    If \cref{alg:iteratedlp} returns some set $G$ on \cref{ln:return-grid}, then there does not exist an exact $k$-query translation-invariant algorithm for searching an $n$-element list. Conversely, if polynomials $q_0, \ldots, q_k$ are returned on \cref{ln:return-polys} then these polynomials witness the existence of an exact $k$-query translation-invariant algorithm for the $n$-element ordered search problem.
\end{prop}

\begin{proof}
    Suppose some set $G$ is returned on \cref{ln:return-grid}. This occurs only if the optimal value of the linear program $\mathcal{L}(G)$ is negative. This in turn implies that there are no nonnegative polynomials satisfying the constraints in \cref{eq:lp-initial,eq:lp-forward-step,eq:lp-final,eq:lp-constant-term,eq:lp-non-negativity} of the linear program $\mathcal{L}(G)$. Since $\mathcal{L}(G)$ is a relaxation of the program in \cref{thm:fggs-pos-polys}, it follows that there does not exist an exact $k$-query translation-invariant algorithm for searching an $n$-element list.
    
    Now suppose that polynomials $q_0, \ldots, q_k$ are returned on \cref{ln:return-polys} of the algorithm. These polynomials constitute a feasible solution of $\mathcal{L}(G)$ for some set $G$, so they satisfy \cref{eq:lp-initial,eq:lp-forward-step,eq:lp-final,eq:lp-constant-term}. The constraints in \cref{eq:lp-initial,eq:lp-forward-step,eq:lp-final,eq:lp-constant-term} imply \cref{eq:fggs-pos-polys-init-state,eq:fggs-pos-polys-forward-step,eq:fggs-pos-polys-final-state,eq:fggs-pos-polys-norm}, respectively. By construction the polynomials are symmetric and real-valued. They are also nonnegative (by \cref{ln:if-m-empty}), as their Chebyshev representations are nonnegative in $[-1,1]$. Thus, the polynomials $q_0, \ldots, q_k$ are a feasible solution for the program of \cref{thm:fggs-pos-polys}, witnessing that there exists a $k$-query exact translation-invariant algorithm for the $n$-element ordered search problem.
\end{proof}

\subsection{Computational Results}

To identify the maximal value of $n$ for which there exists an exact, translation-invariant algorithm for the $n$-element ordered search problem using $k = 5$ queries, we implemented \cref{alg:iteratedlp} and wrapped it with a binary search over $n$. Our program was built in Python using the CVXPY and MOSEK libraries for modeling and solving the linear program $\mathcal{L}(G)$ \cite{diamond2016cvxpy,agrawal2018rewriting,mosek}. 

Our algorithm ran on a machine with an Intel Xeon Silver 4216 2.10 GHz processor and 128 GB RAM. Our approach needed some intermittent manual oversight, primarily due to the scheduling limitations of the machine. In total the search took about a month (including downtime) and identified $n = 7265$ as the maximal instance size searchable with $5$ queries. Our program returned both the witnessing polynomials $q_0, \ldots, q_5$ (depicted in \cref{fig:poly_coeffs}) for $n = 7265$, as well as a set $G$ for which the linear program $\mathcal{L}(G)$ (with $n = 7266$) has a negative optimal value $\beta^* < 0$, certifying the non-existence of 5-query algorithms for searching a list of 7266 elements. By \cref{cor:monotonicity}, it follows that there is also no 5-query algorithm for searching a list of any length more than 7266.

Our code and data are available on GitHub at \url{https://github.com/jacarolan/ordered_search_public}.

\begin{figure}
    \centering
    \includesvg[width=0.96\linewidth]{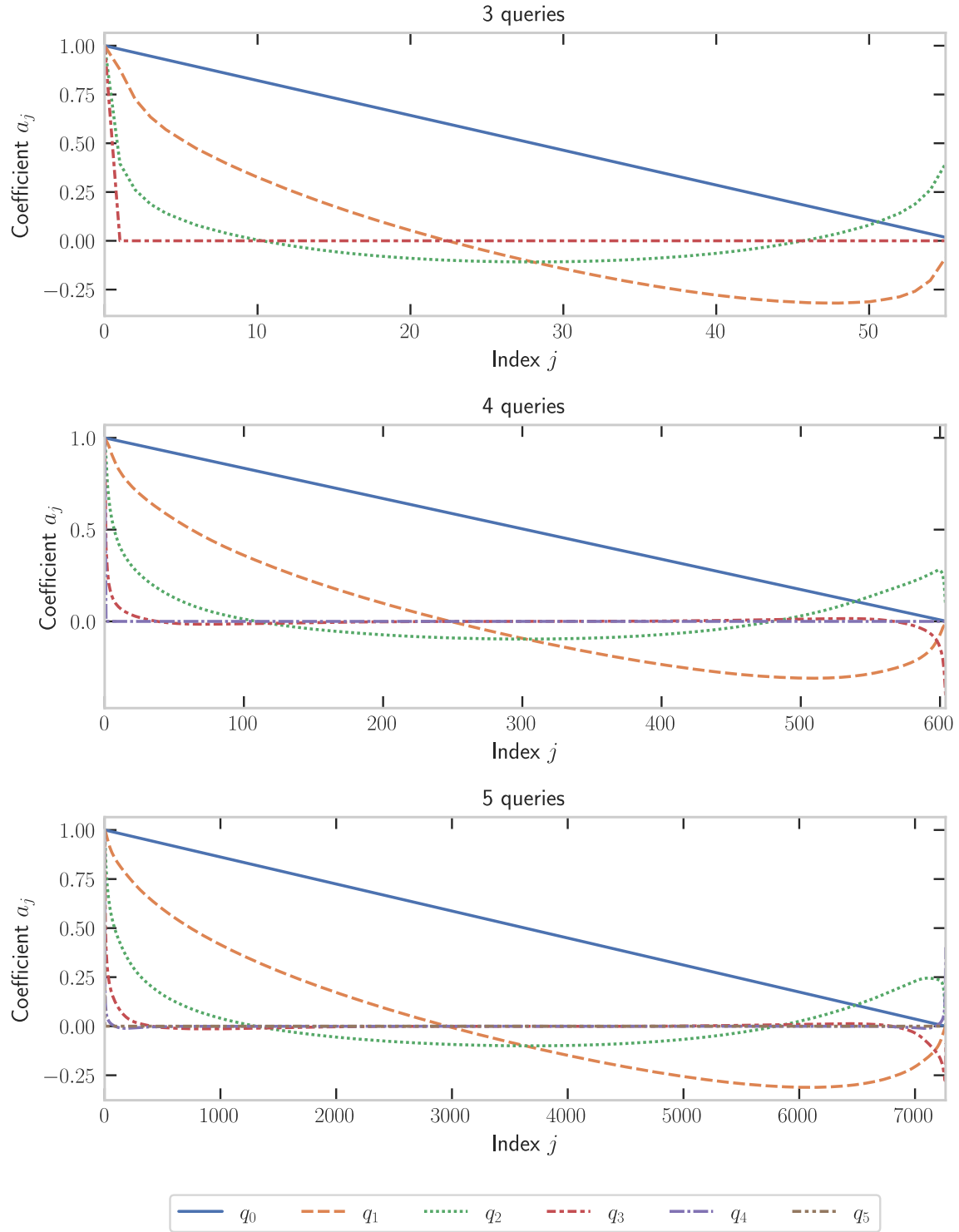}
    \caption{Coefficients of the Laurent polynomials corresponding to exact 3-, 4-, and 5-query algorithms. The corresponding instances sizes are 56, 605, and 7265, respectively. For each polynomial $q_t(z) = \sum_{j=-n}^{n} a^{(t)}_j z^j$ we only plot the coefficients $a_{j}^{(t)}$ for nonnegative indices $j$, as $a^{(t)}_j = a^{(t)}_{-j}$ for each of these polynomials.}\label{fig:poly_coeffs}
\end{figure}

\section{Future Work}

In this paper we showed that translation-invariant algorithms are optimal among all quantum query algorithms for the ordered search problem, for all choices of the error parameter $\varepsilon \geq 0$. We also improved the best upper bound for exact algorithms from (approximately) $0.433 \log_{2}{n}$ to $0.390 \log_{2}{n}$ by exhibiting a 5-query algorithm that can search a list of 7265 elements exactly. 

We mention some potential avenues for further research. First, while we narrowed the gap between upper and lower bounds for ordered search, the precise quantum query complexity of ordered search remains unknown.  One could hope to give a better algorithm, or prove a stronger lower bound. On the algorithmic side, our results show that translation invariance can be taken without loss of generality. This motivates developing tools for better understanding translation-invariant algorithms. Alternatively, one could study bounded- or zero-error\footnote{Zero error here refers to Las Vegas algorithms, which succeed with certainty after a randomized number of queries.} algorithms through the lens of translation invariance. For instance, one could compute the translation invariant analog of the algorithm developed in \cite{bh07}.

On the lower bounds side, we know that the adversary method cannot yield better lower bounds than what is already known \cite{ct08}. However, the error dependence of quantum lower bounds for ordered search has an unusual scaling. In particular, for general algorithms with success probability less than $0.89$, the best known lower bound is $\frac{1}{12} \log_{2}{n}$ \cite{amb99}, which is a constant factor weaker than the $\frac{1}{\pi}\ln{n}$ bound for exact algorithms. Classically, a $\log_{2}{n}-O(1)$ bound applies to any algorithm with constant success probability. It would be interesting to sharpen the error dependence of quantum lower bounds for ordered search.

\section*{Acknowledgments}

JC and AMC acknowledge support from the U.S.\ Department of Energy (grant DE-SC0020264). AMC and MKD acknowledge support from the U.S.\ Department of Energy, Office of Science, Accelerated Research in Quantum Computing, Fundamental Algorithmic Research toward Quantum Utility (FAR-Qu). LS acknowledges the support of the Hartree postdoctoral fellowship and the Joint Center for Quantum Information and Computer Science, where much of this work was done.

\appendix
\section{Uniqueness of the Initial Matrix}

\begin{lem}\label{lem:osp_bss_simplified_first_constraint_unique_sln}
    The semidefinite program 
    \begin{align}
        &\tr_l{A} + \tr_{n-l}{A} = 1 \quad (0 \leq l < n) \label{eq:bss_osp_simplified_initial_constraint_duplicate} \\ 
        &A \in \psd{n}
    \end{align}
    has the unique solution $A = \frac{1}{n} J_n$.
\end{lem}

\begin{proof}
    Since $A$ is positive semidefinite, all of its diagonal entries (being principal minors) are nonnegative. For distinct $i, j \in \{0, \ldots, n-1\}$, by considering the principal minor 
    \begin{align}
    0 \leq 
        \begin{vmatrix}
            a_{ii} & a_{ij} \\ 
            a_{ji} & a_{jj}
        \end{vmatrix} = a_{ii} a_{jj} - a_{ij}^2,
    \end{align}
    we see that $a_{ij} \leq \sqrt{a_{ii} a_{jj}}$. By the AM-GM inequality, 
    \begin{align}
        a_{ij} &\leq \sqrt{a_{ii} a_{jj}} \label[ineq]{eq:psd_mx_entry_upper_bound} \\ 
        &\leq \frac{a_{ii} + a_{jj}}{2}. \label[ineq]{eq:am_gm}
    \end{align}
    Then by \cref{eq:bss_osp_simplified_initial_constraint_duplicate}, 
    \begin{align}
        1 &= \tr_l{A} + \tr_{n-l}{A} \\ 
        &= \sum_{j=0}^{n-1-l} a_{j,l+j} + \sum_{j=0}^{l-1} a_{j,n-l+j} \\ 
        &\leq \frac{1}{2} \Biggl(\sum_{j=0}^{n-l-1} (a_{jj} + a_{l+j,l+j}) + \sum_{j=0}^{l-1} (a_{jj} + a_{n-l+j,n-l+j}) \Biggr) \\
        &= \tr{A} \\
        &= 1.
    \end{align}
    We conclude that \cref{eq:psd_mx_entry_upper_bound,eq:am_gm} hold with equality for all distinct $i, j\in \{0,\ldots,n-1\}$. Therefore $a_{ii} = a_{jj} = 1/n$, and in particular $a_{ij} = 1/n$ as well.
\end{proof}

\bibliography{ref}
\bibliographystyle{alphaurl}

\end{document}